%% file: main.tex
\documentclass[sigconf, nonacm]{acmart}
\AtBeginDocument{%
	\providecommand\BibTeX{{%
			Bib\TeX}}}

\usepackage{wrapfig}
\usepackage{graphicx}
\usepackage{textcomp}
\usepackage{xcolor}
\usepackage[colorinlistoftodos,prependcaption,textsize=footnotesize]{todonotes}

\usepackage{enumitem}
\input{solidity-highlight.tex}
\input{circom-highlight.tex}
\input{json-highlight.tex}
\usepackage{graphicx}
\usepackage{multirow}
\usepackage{xurl}
\usepackage{hyperref}
\usepackage{arydshln}
\usepackage{pifont}
\usepackage{xcolor, colortbl}
\usepackage[linesnumbered,ruled,vlined]{algorithm2e}
\usepackage{algpseudocode}
\usepackage[flushleft, para]{threeparttable}
\usepackage{boldline}
\usepackage{fancyvrb}
\usepackage{soul}
\usepackage{adjustbox}
\usepackage[operators,sets,probability]{cryptocode}
\usepackage{makecell}
\usepackage[dvipsnames]{xcolor}

\hypersetup{
  colorlinks   = true, 
  urlcolor     = blue, 
  linkcolor    = blue, 
  citecolor   = red 
}

\usepackage{xspace}
\newcommand{\kibitz}[2]{\ifnum\Comments=1\textcolor{#1}{#2}\fi}

\newcount\Comments
\def\BibTeX{{\rm B\kern-.05em{\sc i\kern-.025em b}\kern-.08em
    T\kern-.1667em\lower.7ex\hbox{E}\kern-.125emX}}

\usepackage{listings}
\usepackage{xcolor}

\definecolor{codegray}{rgb}{0.95,0.95,0.95}
\definecolor{codeblue}{rgb}{0.0,0.0,0.6}
\definecolor{codegreen}{rgb}{0,0.5,0}
\definecolor{codered}{rgb}{0.6,0,0}

\usetikzlibrary{positioning, arrows.meta, shapes.geometric, calc, shadows}

\definecolor{compGreen}{RGB}{50, 160, 60} 
\definecolor{failRed}{RGB}{200, 50, 50}   

\usetikzlibrary{positioning, arrows.meta, shapes.geometric, shadows, calc, backgrounds}
\usetikzlibrary{positioning, arrows.meta, shapes.geometric, calc, fit, backgrounds}
\usetikzlibrary{shapes.geometric, arrows.meta, positioning, shadows, backgrounds, fit, calc}
\definecolor{physGray}{RGB}{150, 150, 160} 
\definecolor{foldGray}{RGB}{240, 240, 245} 
\definecolor{accColor}{RGB}{100, 100, 100} 
\definecolor{ptBlue}{RGB}{70, 130, 180}   
\definecolor{stGreen}{RGB}{50, 160, 60}   

\lstdefinestyle{terminalStyle}{
    backgroundcolor=\color{codegray},   
    commentstyle=\color{codegreen},
    keywordstyle=\color{codeblue}\bfseries,
    stringstyle=\color{codered},
    basicstyle=\ttfamily\footnotesize, 
    breakatwhitespace=false,         
    breaklines=true,                 
    captionpos=b,                    
    keepspaces=true,                 
    frame=single,                    
    framesep=3pt,
    rulecolor=\color{black!30},      
    showspaces=false,                
    showstringspaces=false,
    showtabs=false,                  
    tabsize=2,
    language=bash,
    morekeywords={python3, bomz, ./run_icse_experiments.sh} 
}

\input{macros}

\begin{document}


\title{%
\veribom: Secure and Verifiable SBOM Sharing Via Zero-Knowledge Proofs
}%
\author{%
Gianpietro Castiglione
}
\email{Gianpietro.Castiglione@newcastle.ac.uk}
\affiliation{%
  \institution{Newcastle University}
  \country{United Kingdom}
}
\author{%
Shahriar Ebrahimi
}
\email{Shahriar.Ebrahimi@newcastle.ac.uk}
\affiliation{%
  \institution{Newcastle University}
  \country{United Kingdom}
}
\author{%
Narges Khakpour
}
\email{Narges.Khakpour@newcastle.ac.uk}
\affiliation{%
  \institution{Newcastle University}
  \country{United Kingdom}
}
\begin{abstract}
A \textit{Software Bill of Materials} (SBOM) is a key component for the transparency of software supply chain; it is a structured inventory of the components, dependencies, and associated metadata of a software artifact. 
However, an SBOM often contain sensitive information that organizations are unwilling to disclose in full to anyone, for two main concerns: technological risks deriving from exposing proprietary dependencies or unpatched vulnerabilities, and business risks, deriving from exposing architectural strategies. 
Therefore, delivering a plaintext SBOM may result in the disruption of the intellectual property of a company.

To address this challenge, we present \veribom, a trustless, selectively disclosed SBOM framework that provides cryptographic verifiability of SBOMs using zero-knowledge proofs.
Within \veribom, third parties can validate specific statements about a delivered software, mainly regarding the authenticity of the dependencies and policy compliance, without inspecting the content of an SBOM. 
Respectively, \veribom\ allows independent third parties to verify if a software contains authentic dependencies distributed by official package managers and that the same dependencies satisfy rigorous policy constraints such as the absence of vulnerable dependencies or the adherence with specific licenses models.

\veribom\ leverages a scalable vector commitment scheme together with folding-based proof aggregation to produce succinct zero-knowledge proofs that attest to security and compliance properties while preserving confidentiality. 
Crucially, the verification process requires no trust in the SBOM publisher beyond the soundness of the underlying primitives, and third parties can independently check proofs against the public cryptographic commitments.

We implement \veribom, analyze its security in the probabilistic polynomial-time (PPT) adversary model, and evaluate its performance on real-world package registries. The results show that our method enables scalable, privacy-preserving, and cryptographically verifiable SBOM sharing and validation.
\end{abstract}

\maketitle


\input{Sec1-Introduction}
\input{Sec6-RelatedWork}
\input{Sec2-Introduction-ZKP}
\input{Sec3-Methodv2}

\input{Sec4-Soundness}

\input{Sec5-Evaluation}

\section{Conclusions}\label{sec:discussion}
We proposed \veribom, a framework for selective disclosure of SBOMs using vector commitments and zero-knowledge proofs.
\veribom\ enables independent third parties to verify the authenticity of artifact dependencies and enforce policy constraints against SBOM metadata in a privacy-preserving manner without revealing the complete SBOM. 
\veribom\ has been implemented as a prototype of the described approach and evaluated its scalability. To address scalability issues, we employed a folding technique for proof generation integrated with a semantic engine for efficient propagation of policy constraints among dependencies.
Our experiments show its effectiveness in reducing overhead.
As future work, we plan to deploy our system in real-world package distribution settings to evaluate its effectiveness under operational conditions. We also aim to evaluate it with other package repositories and improve our techniques for efficient incremental SBOM updates, as well as exploring enforcing other policies.


\section*{Acknowledgment}
This work is supported by the UK Engineering and Physical Sciences Research Council (EPSRC) through grants EP/X037274/1.

\newpage

\bibliographystyle{ieeetr}
\bibliography{ref}

%

\end{document}

%% file: solidity-highlight.tex
\usepackage{listings, xcolor}

\definecolor{verylightgray}{rgb}{.97,.97,.97}

\lstdefinelanguage{Solidity}{
	keywords=[1]{anonymous, assembly, assert, balance, break, call, callcode, case, catch, class, constant, continue, constructor, contract, debugger, default, delegatecall, delete, do, else, emit, event, experimental, export, external, false, finally, for, function, gas, if, implements, import, in, indexed, instanceof, interface, internal, is, length, library, log0, log1, log2, log3, log4, memory, modifier, new, payable, pragma, private, protected, public, pure, push, require, return, returns, revert, selfdestruct, send, solidity, storage, struct, suicide, super, switch, then, this, throw, transfer, true, try, typeof, using, value, view, while, with, addmod, ecrecover, keccak256, mulmod, ripemd160, sha256, sha3}, 
	keywordstyle=[1]\color{blue}\bfseries,
	keywords=[2]{address, bool, byte, bytes, bytes1, bytes2, bytes3, bytes4, bytes5, bytes6, bytes7, bytes8, bytes9, bytes10, bytes11, bytes12, bytes13, bytes14, bytes15, bytes16, bytes17, bytes18, bytes19, bytes20, bytes21, bytes22, bytes23, bytes24, bytes25, bytes26, bytes27, bytes28, bytes29, bytes30, bytes31, bytes32, enum, int, int8, int16, int24, int32, int40, int48, int56, int64, int72, int80, int88, int96, int104, int112, int120, int128, int136, int144, int152, int160, int168, int176, int184, int192, int200, int208, int216, int224, int232, int240, int248, int256, mapping, string, uint, uint8, uint16, uint24, uint32, uint40, uint48, uint56, uint64, uint72, uint80, uint88, uint96, uint104, uint112, uint120, uint128, uint136, uint144, uint152, uint160, uint168, uint176, uint184, uint192, uint200, uint208, uint216, uint224, uint232, uint240, uint248, uint256, var, void, ether, finney, szabo, wei, days, hours, minutes, seconds, weeks, years},	
	keywordstyle=[2]\color{teal}\bfseries,
	keywords=[3]{block, blockhash, coinbase, difficulty, gaslimit, number, timestamp, msg, data, gas, sender, sig, value, now, tx, gasprice, origin},	
	keywordstyle=[3]\color{violet}\bfseries,
	identifierstyle=\color{black},
	sensitive=true,
	comment=[l]{//},
	morecomment=[s]{/*}{*/},
	commentstyle=\color{gray}\ttfamily,
	stringstyle=\color{red}\ttfamily,
	morestring=[b]',
	morestring=[b]"
}

%% file: circom-highlight.tex
\usepackage{listings, xcolor}

\definecolor{verylightgray}{rgb}{.97,.97,.97}

\lstdefinelanguage{Circom}{
	keywords=[1]{public, output, for, input, template, component}, 
	keywordstyle=[1]\color{blue},
	keywords=[2]{signal,address, bool, byte, bytes, bytes1, bytes2, bytes3, bytes4, bytes5, bytes6, bytes7, bytes8, bytes9, bytes10, bytes11, bytes12, bytes13, bytes14, bytes15, bytes16, bytes17, bytes18, bytes19, bytes20, bytes21, bytes22, bytes23, bytes24, bytes25, bytes26, bytes27, bytes28, bytes29, bytes30, bytes31, bytes32, enum, int, int8, int16, int24, int32, int40, int48, int56, int64, int72, int80, int88, int96, int104, int112, int120, int128, int136, int144, int152, int160, int168, int176, int184, int192, int200, int208, int216, int224, int232, int240, int248, int256, mapping, string, uint, uint8, uint16, uint24, uint32, uint40, uint48, uint56, uint64, uint72, uint80, uint88, uint96, uint104, uint112, uint120, uint128, uint136, uint144, uint152, uint160, uint168, uint176, uint184, uint192, uint200, uint208, uint216, uint224, uint232, uint240, uint248, uint256, var, void, ether, finney, szabo, wei, days, hours, minutes, seconds, weeks, years},	
	keywordstyle=[2]\color{teal},
	keywords=[3]{block, blockhash, coinbase, difficulty, gaslimit, number, timestamp, msg, data, gas, sender, sig, value, now, tx, gasprice, origin},	
	keywordstyle=[3]\color{violet},
        otherkeywords = {;,\[,\],*,+,-,--,<<,>>,.,<==,===,?,:,=,<,++},
        keywordstyle=[4]\color{red},
        keywords=[4]{;,<<,>>,?,:,<==,.,===,=,<,-,--,*,+,++},
	keywordstyle=[5]\color{brown},
        keywords=[5]{\[,\]},
        identifierstyle=\color{black},
	sensitive=true,
	comment=[l]{//},
	morecomment=[s]{/*}{*/},
	commentstyle=\color{gray},
	stringstyle=\color{red},
	morestring=[b]',
	morestring=[b]"
}

%% file: json-highlight.tex
\usepackage{listings, xcolor}

\definecolor{delim}{RGB}{20,105,176}
\definecolor{numb}{RGB}{106, 109, 32}
\definecolor{string}{rgb}{0.64,0.08,0.08}

\lstdefinelanguage{json}{
    showspaces=false,
    showtabs=false,
    breaklines=true,
    postbreak=\raisebox{0ex}[0ex][0ex]{\ensuremath{\color{gray}\hookrightarrow\space}},
    breakatwhitespace=true,
    basicstyle=\ttfamily\small,
    upquote=true,
    morestring=[b]",
    stringstyle=\color{string},
    literate=
     *{0}{{{\color{numb}0}}}{1}
      {1}{{{\color{numb}1}}}{1}
      {2}{{{\color{numb}2}}}{1}
      {3}{{{\color{numb}3}}}{1}
      {4}{{{\color{numb}4}}}{1}
      {5}{{{\color{numb}5}}}{1}
      {6}{{{\color{numb}6}}}{1}
      {7}{{{\color{numb}7}}}{1}
      {8}{{{\color{numb}8}}}{1}
      {9}{{{\color{numb}9}}}{1}
      {\{}{{{\color{delim}{\{}}}}{1}
      {\}}{{{\color{delim}{\}}}}}{1}
      {[}{{{\color{delim}{[}}}}{1}
      {]}{{{\color{delim}{]}}}}{1},
}

%% file: macros.tex
\newcommand{\auditor}{auditor\xspace}
\newcommand{\prover}{prover\xspace}

\newcommand{\commitment}{\ensuremath{\mathcal{C}}\xspace}
\newcommand{\commit}{ \ensuremath{\mathsf{commit}}\xspace}

\newcommand{\digest}{\ensuremath{g}\xspace}
\newcommand{\False}{ \ensuremath{\bot}\xspace}
\newcommand{\hashFun}{\ensuremath{H}\xspace}

\newcommand{\keygen}{ \ensuremath{\mathsf{key\_gen}}}
\newcommand{\leaf}{ \ensuremath{\ell}\xspace}

\newcommand{\MerklePath}[1]{\ensuremath{\pi_{#1}}\xspace}
\newcommand{\mlanguage}{\ensuremath{\mathcal{L}}\xspace}
\newcommand{\mmessage}{\ensuremath{m}\xspace}
\newcommand{\mproof}[1]{\ensuremath{\MerklePath {#1}}\xspace}

\newcommand{\mvector}{\ensuremath{v}\xspace}
\newcommand{\open}{ \ensuremath{\mathsf{open}}}
\newcommand{\opening}{ \ensuremath{\mathcal{O}}}
\newcommand{\packagesDB}{ \ensuremath{\mathcal{P}}}
\newcommand{\package}{ \ensuremath{p}\xspace}

\newcommand{\prove}{ \ensuremath{\mathsf{prove}}}
\newcommand{\publicParams}{\ensuremath{{\lambda_p}}\xspace}

\newcommand{\relation}{\ensuremath{\mathcal{R}}\xspace}

\newcommand{\securityParam}{\ensuremath{{{\lambda_s}}}\xspace}
\newcommand{\setup}{ \ensuremath{\mathsf{setup}}\xspace}

\newcommand{\structuredHashFunc}{ \ensuremath{{\mathcal{H}}}\xspace}
\newcommand{\PT}{\ensuremath{PT}\xspace}
\newcommand{\ST}{\ensuremath{ST}\xspace}
\newcommand{\VC}{\ensuremath{VC}\xspace}
\newcommand{\PPE}{\ensuremath{PPE}\xspace}
\newcommand{\True}{ \ensuremath{\top}\xspace}

\newcommand{\verify}{ \ensuremath{\mathsf{verify\_commit}}\xspace}
\newcommand{\verifyZKP}{ \ensuremath{\mathsf{verify\_zkp\_proof}}\xspace}

\newcommand{\zkpproof}[2]{ \ensuremath{{\eta_{#2}^{#1}}}\xspace}
\newcommand{\attacker}{\ensuremath{\mathcal{A}}\xspace}
\newcommand{\tattacker}{\ensuremath{\mathcal{A}_{T}}\xspace}
\newcommand{\cattacker}{\ensuremath{\mathcal{A}_{C}}\xspace}
\newcommand{\mattacker}{\ensuremath{\mathcal{A}_{V}}\xspace}

\newcommand{\veribom}{\textsc{VeriSBOM}\xspace}

\newcommand{\STAB}[2]{\begin{tabular}{@{}#1@{}}#2\end{tabular}}
\RestyleAlgo{ruled}

\renewcommand{\arraystretch}{1.4}

\newcommand\lstresetnumber{\global\let\thelstnumber=\orig@lstnumber}

\theoremstyle{definition}
\newtheorem{definition}{Definition}[section]
\newtheorem{theorem}{Theorem}

\newcommand{\dontknow}[1]{\textcolor{blue!70!green}{} }

\newcommand{\leaveout}[1]{\textcolor{blue!70!green}{} }

%% file: Sec1-Introduction.tex
\section{Introduction}


Software supply chain attacks target the trust relationships and processes involved in software development and distribution. These attacks exploit vulnerabilities in third-party dependencies, weaknesses in the build and deployment infrastructure, or human factors such as credential theft and social engineering~\cite{williams2025research}. Attackers may inject malicious code into source repositories, compromise widely used packages through project hijacking, or tamper with the build process to insert backdoors~\cite{h2018eventstream,gruhn2013security,wu2021feasibility}.
To mitigate such threats, several techniques have been proposed, including detecting vulnerabilities in code dependencies~\cite{ccs/DuanBXKL17,icse/WooPKLO21,uss/Xiao0YXYLL0HZS20,sp/KimWLO17,vu2021lastpymile}, securing package repositories~\cite{KuppusamyDC17,kuppusamy2016diplomat}, and protecting the build process through verifiable builds, cryptographic signing, and strict access controls~\cite{TorresAriasAKC19,ferraiuolo2022policy} (Ladisa et al. published a survey on the topic \cite{ladisa2023sok}). 
Beyond these technical measures, \emph{transparency} has emerged as a critical principle for reinforcing trust in the software supply chain. By improving visibility and accountability, transparency mechanisms help stakeholders detect and mitigate hidden risks. One widely adopted approach to achieve this is the Software Bill of Materials (SBOM)~\cite{ladisa2023sok,xia2023empirical}, which provides an inventory of all software components, including third-party libraries and their versions. SBOMs enable security risk assessments, vulnerability identification, and regulatory compliance~\cite{xia2023empirical}. For instance, the U.S. Office of Management and Budget now mandates SBOMs for federal software procurement, following NIST SSDF guidelines.  

\paragraph{Challenges.}
Producing an SBOM is only the first step: it must be securely shared, verified, and integrated into downstream workflows to ensure compliance with security, quality, and licensing requirements. Achieving this raises significant challenges in the secure sharing and validation of SBOMs that must be addressed for widespread adoption. SBOMs may expose confidential or proprietary information that adversaries could exploit to locate and target vulnerable components, or that competitors could infer for architectural or strategic insight~\cite{zahan2023software}. Hence, concerns about sensitive information disclosure remain a major barrier to broader SBOM adoption~\cite{hendrick2022software,xia2023empirical}. To mitigate these risks, SBOM sharing practices must support selective information sharing and access control to balance transparency with protection of sensitive data~\cite{xia2023empirical}. Among the proposed conceptual frameworks for selective SBOM sharing ~\cite{song2025bc,xia2024trust,ishgair2026trustworthyconfidentialsbomexchange}, no practical solution yet enables privacy-preserving compliance verification in real-world deployments (\textbf{C1}).

Existing SBOM validation often assumes access to source code or build environments to verify SBOM data (e.g., dependencies), which is impractical in closed-source or proprietary settings. Even when SBOMs are available, users may lack the expertise or tooling to interpret them and enforce security or policy requirements, particularly in rapidly evolving software systems where dependencies change frequently. At the same time, selective (privacy-preserving) SBOM disclosure creates a verification gap: without full visibility, a publisher could omit a vulnerable or disallowed dependency, alter version metadata, or make unverifiable policy claims. Thus, consumers need guarantees that disclosed dependencies are authentic, that hidden entries do not violate stated security or policy conditions, and that verification remains efficient as dependencies evolve. These risks make \emph{verifiability} crucial: we need trustless, proof-based validation to ensure that disclosure does not become a vector for concealment or manipulation (\textbf{C2}).

\paragraph{Solution.}
To address challenges C1 and C2, we propose \veribom, a novel approach for secure and verifiable SBOM sharing that leverages two cryptographic primitives: vector commitment (VC) schemes~\cite{vc2,vc3} and zero-knowledge proof systems (ZKPs)~\cite{goldwasser1985knowledge}. 
We use VC schemes to build a \textit{Dual-Tree Architecture}: a \textit{Package Tree} obtained from a \textit{package manager} to commit the inclusion and authenticity of dependencies within a package repository and a parallel, isomorphic \textit{Shadow Tree} obtained from an independent \textit{auditor} to commit the compliance of the packages to specific policy constrains.
These commitments are made publicly available as authoritative \textit{roots}, and on top of them, \veribom~ employs ZKPs, a cryptographic primitive where a \emph{prover} generates a proof about an SBOM, which a \textit{verifier} can use to validate the claim against the public roots, without revealing any additional information about the SBOM itself (addressing \textbf{C1}).
The proof shows that the dependencies in an artifact's SBOM are authentic and satisfy policy constraints (addressing \textbf{C2}).

A key challenge in designing ZKP proofs is their potentially high computational cost, including memory consumption and the time required for proof generation and verification. This challenge is particularly significant when validating SBOMs of large software systems with many dependencies, as each dependency must be individually checked against the same set of constraints in our approach. 
To address this, we leverage folding schemes~\cite{nova-paper,nova}, a special class of ZKPs that aggregate multiple executions of a \textit{verification circuit} into a constant-sized proof. 

\paragraph{Contributions.} Within this work, the key contributions are: 

\begin{itemize}
	\item \emph{Privacy-preserving properties verification.} 
    We propose an approach for sharing SBOMs using ZKPs, which enables the clients to verify  policies about the dependencies without exposing sensitive information.

    \item \emph{Security analysis.}
    We analyze the security of our approach in different scenarios, against a {probabilistic polynomial-time} adversary with varying objectives.
    
    \item \textit{Scalability.} 
    To address scalability issues, we leveraging \textit{recursive folding-based schemes} for ZKP-proof generation as well as a \textit{policy engine}, to reduce the verification complexity.
    	
	\item \emph{Implementation and evaluation.} 
    We implement \veribom\ as a prototype integrating the designed architecture and conduct experiments demonstrating it effectiveness and scalability to real-world package repositories.
\end{itemize}

The remainder of this paper is organized as follows. Section~\ref{sec:related} compares our approach with related work. Section~\ref{sec:background} provides the necessary background to our work. Sections \ref{sec:overview} and \ref{sec:protocol-overview} introduce and describe our approach, respectively. 
In Section~\ref{sec:security-analysis}, we describe some properties of our approach.
In Section~\ref{sec:evaluation} we discuss the implementation and some experiments. Section~\ref{sec:discussion} concludes the paper and outlines directions for future research.

%% file: Sec6-RelatedWork.tex
\section{Related Work}\label{sec:related}
\paragraph{SBOM Security}
In a recent work, Ishgair et al. have proposed Petra \cite{ishgair2026trustworthyconfidentialsbomexchange} for addressing confidentiality and selective disclosure of SBOMs by using Ciphertext-Policy Attribute-Based Encryption (CP-ABE).
Although Petra seems to answer the challenges we highlight, our solution differs significantly from Petra, both in its objectives and in the tools used. Such fundamental differences are illustrated in Table \ref{tab:comparison} and generally explained in the following.

\begin{table}[h]
    \scriptsize
    \centering
    \caption{\textbf{Comparison between Petra and \veribom.} While Petra relies on decryption for verification (exposing data), \veribom\ enables \textit{blind verification} of properties without key distribution.}
    \label{tab:comparison}
    \renewcommand{\arraystretch}{1.2} 
    \begin{tabular}{@{} l p{3.5cm} p{2.5cm} @{}}
        \toprule
        \textbf{Feature} & \textbf{Petra} & \textbf{\veribom} \\
        \midrule
        \textbf{Privacy Mechanism} & Attribute-Based Encryption (CP-ABE) & Zero-Knowledge Proofs (ZK-SNARKs/Folding) \\
        
        \textbf{Verification Scope} & \textit{Structural Integrity} \newline (Is the redacted file authentic?) & \textit{Semantic Compliance} \newline (Does the hidden file satisfy security policies?) \\
        
        \textbf{Verifier Visibility} & \textbf{Partial Disclosure} \newline (Must decrypt data to verify it) & \textbf{Zero Disclosure} \newline (Blind verification on hidden data) \\
               
        \textbf{Key Management} & Requires per-user secret keys & Keyless Public Proofs \\
        
        \textbf{Interaction} & Interactive \newline (Verifier requests keys from KMS) & Non-Interactive \newline (Verifier checks public knowledge) \\
        \bottomrule
    \end{tabular}
\end{table}

Petra effectively enforces granular access control to specific parts of an SBOM that a software producer has redacted. To access the confidential parts of the SBOM, the verifier must possess the necessary key and attributes to decrypt and view the underlying metadata in plaintext. 
The distribution of keys via the Key Management System (KMS) may create a 'key management bottleneck': KMS require to be continuously online to issue credentials for every potential verifier in the global supply chain.
Conversely, our framework introduces \textit{blind verifiability}. 
By leveraging zero-knowledge proofs, we allow any verifier to cryptographically confirm that obfuscated SBOM parts satisfy specific policies without ever revealing the component details. 
Furthermore, while Petra relies on an interactive key management service to distribute credentials, our solution is non-interactive and keyless, relying solely on public knowledge necessary for the verification process.

Song et al.~\cite{song2025bc} and Xia et al.~\cite{xia2024trust} propose frameworks based on blockchain for managing SBOMs, but they both remain as conceptual frameworks, without providing any concrete solutions. Song et al. provide a decentralized and transparent approach supporting dependency queries and vulnerability alert propagation but does not concern with secure sharing of SBOM. 
Xia et al. ensure integrity and selective disclosure via verifiable credentials but incur the cost and complexity of blockchain consensus. 
Other approaches designate different solutions for SBOMs, not scalable at all \cite{10.1145/3664476.3669975}, \cite{10.1145/3654442}, \cite{linuxfoundation2022sbom}.
In contrast, we  implemented our scheme, using vector commitments and zero-knowledge proofs to enable selective disclosure and authenticity verification of dependencies in a confidentiality-preserving manner, without expensive consensus, and we prove its soundness.


\paragraph{Detecting Malicious Packages}
Vulnerability detection in software packages typically follows two approaches: version-based methods that check if a code version has known vulnerabilities~\cite{ccs/DuanBXKL17,icse/WooPKLO21}, and code-based methods that match code fragments to vulnerable patterns~\cite{uss/Xiao0YXYLL0HZS20,sp/KimWLO17}. V1SCAN~\cite{uss/WooCLO23} combines both using classification for C/C++ programs. Beyond detection, works like \textsc{Centris}~\cite{icse/WooPKLO21} and VDiOS~\cite{icse/ReidJM22} analyze modified or reused open source code to trace vulnerabilities, while Zahan et al.~\cite{zahan2022weak} detect npm supply chain weaknesses via metadata analysis. Software Composition Analysis (SCA) also identifies third-party components and risks but lacks end-to-end supply chain integrity.
Furthermore, several works address package manager security; for instance, \cite{KuppusamyDC17} proposes a method to prevent rollback attacks, which revert projects to vulnerable versions.
Our approach complements these efforts. While prior work focuses on vulnerability detection and propagation analysis via code analysis, we aim to ensure software supply chain integrity by enabling secure SBOM sharing with verifiable claims, allowing the authenticity and policy compliance of dependencies to be validated.

\paragraph{Security of Software Production Process}
Several approaches focus on improving software supply chain security by securing the software production process. Among these, Torres-Arias et al.~\cite{TorresAriasAKC19} propose \textit{in-toto}, a framework that enforces supply chain integrity through cryptographically signed metadata generated at each production step. 
To further ensure accountability, \textit{PolyLog}~\cite{ferraiuolo2022policy} introduces a framework to improve auditability of the software supply chain through \textit{Policy Transparency}, where claims about the artifact are expressed using an authorization logic and published in transparency logs, facilitating software auditing and helping to identify potential abuses. While these approaches focus on securing the software production phase, which may include generating an SBOM, our work emphasizes the secure sharing and validation of SBOMs.

\paragraph{Cryptography for Software Supply Chain}
Cryptographic techniques have long been applied to secure software artifacts and mitigate supply chain attacks. Package and repository signing frameworks such as Sigstore~\cite{newman2022sigstore} lower adoption barriers through developer-friendly signing infrastructure, while Speranza~\cite{merrill2023speranza} leverages zero-knowledge co-commitments to ensure authenticity without exposing signer identities. Diplomat~\cite{kuppusamy2016diplomat} improves repository resilience to key compromise via flexible trust delegations, LastPyMile~\cite{vu2021lastpymile} detects inconsistencies between source and distributed packages through artifact signing and scraping, and ShIFt~\cite{singi2019shift} builds composite software identities from code, dependencies, and configurations to detect tampering in versions. While these approaches provide authenticity and integrity, or signer anonymity, they  assume full visibility into artifacts which limits their applicability in privacy-sensitive settings. In contrast, our work addresses the challenge of {selective SBOM disclosure} to make hidden components remain verifiable, a privacy-preserving verification problem not covered by prior solutions.

%% file: Sec2-Introduction-ZKP.tex
\section{Preliminaries}\label{sec:background}


\subsection{Vector Commitment Schemes}
\emph{Commitment schemes} are fundamental building blocks in cryptographic protocols that enable a party (called the \emph{committer} or \emph{prover}) to generate and publish a value $\digest = \commitment{}(\mmessage, \publicParams)$, called a \emph{digest}, that \emph{binds} them to an input \mmessage that should be hidden from others. Binding means that once the digest is created, the prover cannot change the input without being detected,
and \commitment is a function to build the digest for \mmessage and some specific {parameters \publicParams.}
At a later stage, the prover may \emph{open} the commitment, i.e., it reveals some (possibly partial) information that is claimed to be associated with the original hidden input~\mmessage corresponding to the published digest \digest. A  \emph{verifier} can then independently check whether the revealed information is consistent with the digest, without having any prior knowledge of the hidden input \mmessage.
This input can take different forms. A \emph{Vector Commitment} (\VC) is a commitment in which the prover commits to a vector $\mvector = [\mvector_1,\ldots,\mvector_n]$ of elements as the input, and the verifier can later verify the value at a specific index of the vector without revealing the entire vector. 


A widely used data structure for constructing and implementing vector commitments is the \emph{Merkle tree}, which is a binary tree built by recursively hashing node values in a bottom-up fashion.
Figure~\ref{fig:commitment} shows an example of a Merkle tree. In this tree, (i) the leaf $\leaf_i$ is the hash of $\mvector_i$ in the vector, i.e., $\ell_i=\structuredHashFunc(\mvector_i)$ where \structuredHashFunc is a hash function; (ii) each internal node's value is the hash of its two child nodes (i.e., {$H(h_{l} \parallel h_{r})$}, where $\parallel$ denotes concatenation, and $h_{l}$ and $h_{r}$ are the hashes of the left and right child nodes, respectively) ; (iii) the digest $\digest$ is the root of the Merkle tree; and (iv) a proof for the statement ``$v_k$ is the $k$-th element of the vector'' consists of the sibling nodes along the path from the $k$-th leaf to the root. 
This path, known as a \emph{Merkle path} (See Figure~\ref{fig:commitment} for an example), has a logarithmic size of $O(\log n)$.

\begin{figure}[t]
	\hspace{-0.18in}
	\includegraphics[scale=0.65]{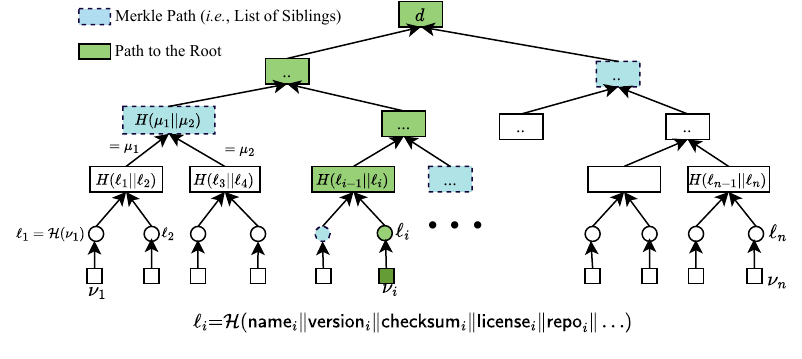}
	\caption{Proposed commitment method for package manager}
	\label{fig:commitment}
\end{figure}
To verify that a leaf $\ell_k$ belongs to a (balanced) Merkle tree with root $\digest$, the verifier is provided with its hash (i.e., $\ell_k$) and a Merkle path $\pi = [h_1, h_2, \dots, h_{\log n}]$, where each $h_i$ is the hash of a sibling node at level $i$ along the path from the leaf to the root, for $1\leq i \leq \log n$.
The verifier performs the following recursive computation to check whether this proof is valid:

\begin{enumerate}
    \item Initialize $x \gets \ell_k$;
    \item For each level $i = 1$ to $\log n$, compute:
    \[
    x :=
    \begin{cases}
        \hashFun(h_{i} \parallel x) & \text{if } h_{i} \text{ is a left child} \\
        \hashFun(x \parallel h_{i}) & \text{if } h_{i} \text{ is a right child}
    \end{cases}
    \]
\end{enumerate}

If the final value of $x$ is \digest, the verifier accepts the proof, confirming that $\ell_k$ is correctly included at position $k$ in the committed vector. Otherwise, the proof is rejected.
This recursive hash computation ensures that no adversary can forge membership for a value not originally committed, assuming that the hash functions \hashFun and \structuredHashFunc are collision-resistant.


%
%
%
%

A \VC scheme is called \emph{{position-binding}} if it guarantees that, after a commitment to a vector has been generated, it is computationally infeasible to produce two distinct valid openings for the same index. Formally, given a commitment $\commitment{}$ to a vector $\mvector$, no adversary can provide two different values $v_i$ and $v'_i$ along with their corresponding valid proofs $\mproof i$ and $\mproof j$
 unless their values are the same, i.e., $v_i = v'_i$.
This property ensures the integrity of each element in the committed vector, making it impossible to alter individual elements of the vector, after the commitment is published.

\subsection{Zero-Knowledge Proof Systems}
Zero-knowledge proof (ZKP) systems are cryptographic protocols that enable one party (the prover) to demonstrate the truthfulness of a statement to another party (the verifier) without revealing any additional data beyond the statement itself in the interaction between the prover and verifier~\cite{goldwasser1985knowledge}. A ZKP must satisfy three fundamental properties. (i) \emph{Completeness}, i.e., if the statement is true, an honest prover will  convince an honest verifier.
(ii) \emph{Soundness} stating that a dishonest prover cannot convince an honest verifier of a false statement.
(iii) \emph{Zero-Knowledge} to ensure that no additional information beyond the validity of the statement is revealed.

\begin{definition}\label{def:snarks}
Let \mlanguage be an NP language (i.e., it can be decided whether a string belongs to this language in non-deterministic polynomial time), and let \relation be a relation consisting of pairs $(u, w)$, where $u \in \mlanguage$ is a statement and $w$ is a witness provided by the prover  to show the truthfulness of the statement $u$.
An \emph{argument system} for $\mathcal{R}$ is a quadruple  $(\setup, \keygen, \prove, \verifyZKP)$ of
 algorithms where:

\begin{description}
    \item $\publicParams \gets {\setup(\lambda)}$ is an algorithm that returns the \emph{public parameters} $\publicParams$ based on the \emph{security parameter} $\lambda$ shared between the prover and verifier. 
    A security parameter shows the security level of a commitment, i.e., parameters chosen to make attacks computationally difficult, and a public parameter includes the setup parameters to guarantee the security level. 

    \item $(pk,vk)\gets\keygen(\publicParams, \relation)$ is an algorithm that returns a proving and verifying key pair $ (pk,vk) $ for a given relation \relation with respect to $ \publicParams $.

    \item $\zkpproof{}{} \gets \prove(pk, u, w)$ to generate a succinct proof $\zkpproof{}{}$ attesting that $(u, w) \in \relation$ for a proving key $pk$, a statement $u$, and a witness $w$.

    \item $b \gets \verifyZKP(vk, u, \zkpproof{}{})$ Given a verification key $vk$, a statement $u$, and a proof $\zkpproof{}{}$, the verifier outputs $b \in \{\True,\False\}$, accepting ($b = \True$) if the proof is valid or rejecting ($b = \False$) otherwise.
\end{description}

An argument system is non-interactive, if prover and verifier do not interact with each other.
A non-interactive argument system is SNARK~\cite{bitansky2012extractable,gennaro2013quadratic,parno2016pinocchio,groth16,plonk}  if it satisfies additional properties: \emph{succinctness} (the proof $\zkpproof{}{}$ is short and can be verified quickly), \emph{soundness}, and optionally \emph{zero-knowledge} (the proof reveals nothing about the witness $w$). 
\emph{SNARKs} focus on efficiency by reducing interaction between the prover and the verifier. Zero-knowledge SNARKs~(zkSNARKs) (e.g., ~\cite{groth16}) are a specific type of SNARKs that have the zero-knowledge property. 
A formal zero-knowledge definition also requires a simulator algorithm, which we omit here for brevity.
\end{definition}

\subsection{Folding Schemes}
Folding-based zkSNARKs are a set of specific cryptographic primitives that achieve incrementally verifiable computation~(IVC) by allowing verification of computations done by repeated application of the same function. More precisely, for a given function $f$, with initial input $z_0$, a folding scheme generates a proof $\Pi_i$ for stating that $z_i=f(z_{i-1})=f^i(z_0)=f(f^{i-1}(z_0))$, given a proof $\Pi_{i-1}$ for stating $z_{i-1}=f^{i-1}(z_0)$. In such a proving system, successful verification asserts with overwhelming probability~($1{-}\mbox{negl}(\lambda)$) that the following condition holds:  $\forall i \in [n], \; z_i = f^i(z_0)$. Nova~\cite{nova-paper} is the first folding scheme that is able to transform the verification of two NP statements into checking a single NP statement. This allows the prover to incrementally prove the correct execution of sequential computations represented by the form $y = F^l(x)$, where $F$ is the computation, $x$ is the initial input, and $l > 0$ represents the number of steps, i.e., application of $F$.

%% file: Sec3-Methodv2.tex
\section{An Overview of \veribom }
\label{sec:overview}
\input{Sec-ThreatModel}
\subsection{Architecture and Context}
\paragraph{Architecture}

A developer releases packages to a \emph{package manager} (e.g., Crates~\cite{cratesio} for Rust libraries), which maintain them in a centralized Package Database (Package DB). 
We construct a {publicly maintained vector commitment, which contains individual commitments for all packages and their respective releases. }
Fig.~\ref{fig:commitment} illustrates the structure of a \VC for a package repository implemented as a Merkle tree in our approach, where each leaf node represents a commitment to a specific version of a package. This means that we associate a leaf to each version of a package, e.g., a package with 100 unique release versions would have 100 unique leaves in the tree.
This process is publicly verifiable and can be preformed independently by any stakeholder in the software supply chain, such as developers, auditors, or third-party parties, without relying on centralized trust. 
Therefore, the final commitment can be maintained in a scalable and decentralized manner,
{where multiple parties can generate the VC in a collaborative way.
}

\begin{lstlisting}[language=json, basicstyle=\ttfamily\small, frame=single, numbers=none, caption={Example SBOM Entry}, label={listing.sbom.data}]
	{
		"name": "serde",
		"version": "1.0.136",
		"dependencies": ["serde_derive"],
		"hash": "abc123...",
		"license": "MIT"
	}
\end{lstlisting}
An SBOM contains various elements depending on the standard, such as component name, publisher, version, dependency list, SBOM author, file name, license information, and timestamp (see Listing~\ref{listing.sbom.data} for an example). Since some of these elements may be sensitive, our approach avoids sharing the full SBOM publicly. Instead, we construct a representation consisting of (i) public metadata with sensitive information concealed (e.g., revealing only a selected subset of dependencies or author details), and (ii) a set of zero-knowledge proofs generated by the \prover. These proofs attest to statements defined by the  consumer, such as the exclusion of known vulnerable dependencies, the inclusion of only approved packages, or compliance with licensing policies. For instance, a consumer can verify that a dependency belongs to an approved list without learning which specific dependency is used. To support complex policy checks that cannot be derived solely from dependency and SBOM metadata, the \prover\ may also rely on an \auditor, a trusted entity with access to the source code to analyze additional constraints and provide the analysis results as input for proof generation. The proofs, are generated using zkSNARK systems such as Spartan~\cite{spartan} or Groth16~\cite{groth16}, and encoded within folding schemes, ensuring validity checking without disclosing private information. 
The code consumer then verifies software dependencies using the public SBOM, the vector commitment of the package repository, and the zero-knowledge proofs provided by the \prover. Leveraging the zero-knowledge property of zkSNARKs, our method guarantees selective disclosure: partial SBOM visibility does not compromise the confidentiality of hidden components or metadata, while still enabling independent verification of compliance with specified constraints.
The roles are the following:
\begin{figure}[t]
    \centering
    \includegraphics[width=0.8\linewidth]{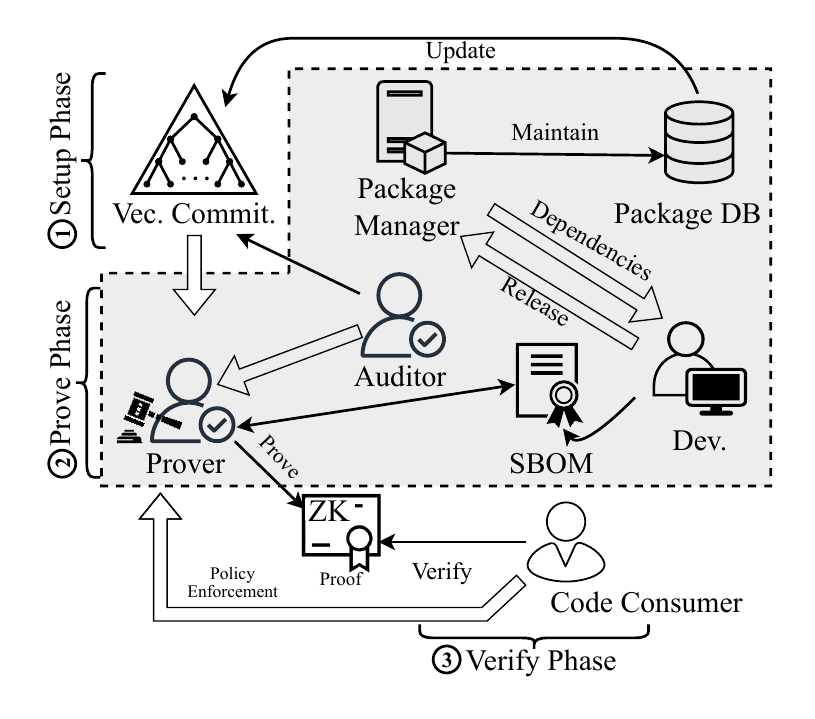}
    \caption{High-level overview of our method. In practice, a single entity may assume multiple roles. 
}  
    \label{fig:sys}
\end{figure}

\begin{itemize}
    \item The \emph{package manager} maintains the package database (or registry) from where the packages/dependencies are downloaded and included in the SBOMs. 
    \item The \emph{auditor} independently checks if the packages satisfy \textit{pre-established} or \textit{on-demand} policy constraints.
    \item The \emph{software developer} develops the software artifacts and generate the SBOM where the software metadata are listed. 
    \item The \emph{prover} generates a privacy-preserving proof against the SBOM. The \prover\ acts as the \textit{SBOM Validator}. The \prover\ can be an external entity (which should be trusted by the software developer) or the software developer itself.
    \item The \emph{consumer} receives the proof (alongside the software artifact) and verifies its validity without learning any data about the SBOM.
\end{itemize}

\paragraph{Trust Assumptions} 
Our model assumes that the package manager and auditor are considered trusted.
All other parties, including the \prover\ are treated as potentially untrusted. 
Although the \prover\ holds the sensitive data, i.e., the SBOM, we assume the worst-case scenario where he deceives the final client.
Hence, the verifier (e.g., the package consumer) must be able to validate SBOM claims only relying on the correctness of the system and not on the honesty of any intermediary.

\paragraph{Security Requirements}
We propose a novel approach that leverages \emph{vector commitments} and \emph{zero-knowledge proof systems} to enable selective disclosure of SBOM entries while allowing third parties to cryptographically verify the authenticity of disclosed components and their compliance with policies, without revealing the full SBOM.
Our method is designed to satisfy the following requirements:

\begin{itemize}
	\item[\textbf{R1}] \emph{Confidentiality:} SBOM sharing should protect sensitive information by allowing selective disclosure of only the data necessary for cryptographic verification.
	
	\item[\textbf{R2}] \emph{Dependency Authenticity:} The mechanism should enable users to verify the authenticity and integrity of disclosed dependencies (e.g., confirming they match the committed data) without relying on trust in the SBOM publisher.
	
	\item[\textbf{R3}] \emph{Policy Compliance:} Users should be able to cryptographically verify that dependencies satisfy security and policy constraints (e.g., licence requirements, version checks, or approved package lists) without requiring full SBOM access or the software's source code.
\end{itemize}

\veribom\ satisfies the security requirements by generating a single Zero-Knowledge proof $\pi$ that attests to the validity of the following assertions:

\begin{itemize}
    \item [\textbf{P1}] \emph{Inclusion Proof (satisfying R1 and R2)}: Users verify that the packages included in an SBOM exist and are immutable entries in the global supply chain state, ensuring they are correctly distributed by the package manager.
    
    \item [\textbf{P2}] \emph{Compliance Proof (satisfying R1 and R3)}: Users verify that the packages included in an SBOM are compliant with a set of policy constraints validated by a security auditor.
\end{itemize}

\section{\veribom\  Description}\label{sec:protocol-overview}

\subsection{Vector Commitment Construction}  
Vector Commitments ($\VC$s) handle the information of the package manager and auditor. A \VC is implemented as a Merkle tree, where each leaf corresponds to a unique entry package or compliance check.
Our \VC structure, like any Merkle tree, is inherently scalable with respect to both proof size and verification complexity. The proof size and the number of hash computations required during verification depend on the height of the tree, which remains small even when accommodating millions of entries. This makes the structure suitable for large-scale software ecosystems, as the height of the tree grows logarithmically with the number of committed packages, i.e., $O(\log n)$. For example, consider a scenario involving one million packages, each with approximately 100 versions, resulting in a total of $2^{27}$ package-version entries. The corresponding Merkle tree would therefore have a height of 27. Assuming a 256-bit (32-byte) hash function (e.g., SHA-256 or Poseidon), the total storage required to represent all internal and leaf nodes would be approximately $2^{27}$ $\times$ 32 bytes, amounting to roughly 4-5 GB. Furthermore, verifying that a particular package-version entry is part of the commitment requires only 26 intermediate hash values (i.e., one per level, excluding the leaf), resulting in a total proof size of 26 $\times$ 32 = 832 bytes, which is less than 1 KB. In addition, the verifier must perform 26 hash computations to validate the proof, which remains computationally efficient even at large scales.

\begin{definition}
Given a database of packages $\packagesDB= \{\package_1, \dots, \package_k\}$, the vector commitment scheme for the package repository is defined as a tuple of algorithms 
$\commitment_\packagesDB = (\setup, \commit, \open, \verify)$ where,

    \begin{itemize}
        \item $\publicParams \gets \setup(\securityParam, D)$: The algorithm initializes the public system parameters $\publicParams$ given a security parameter $\securityParam$ and the  Merkle tree depth $D$, which defines the vector commitment size $N = 2^D$.
        The parameters $\publicParams$ include:
        \begin{itemize}
            \item The specification of the cryptographic hash function.
            \item The fixed tree depth $D$ that ensures a deterministic structure for package indexing.
        \end{itemize}
    
        \item $\digest \gets \commit(\publicParams, \packagesDB)$: The algorithm constructs the commitment to the global state, e.g., the package repository. 
        The algorithm, for each package 
        computes the leaf value and maps it to a specific index \textit{idx}.
        It returns the Merkle Root $\digest$, which serves as the succinct commitment.
    
        \item $\opening \gets \open(idx, \leaf_i, \package_i, \digest)$: Given a package $\package_i$ mapped to index $idx$, this algorithm generates an \textit{opening} \textbf{(P1)} $\opening$. This consists of:
        \begin{itemize}
            \item The Merkle Path (siblings) from the leaf at $idx$ to the root $\digest$.
            \item The pre-image verification data required to reconstruct $\leaf_i$.
        \end{itemize}
    
        \item $b \gets \verify(\digest, \leaf_i, idx, \opening)$: The verification algorithm outputs $b = \True$ if and only if recomputing the Merkle root using $\leaf_i$, the index $idx$, and the path in $\opening$ results in a value equal to $\digest$. Otherwise, it outputs $b = \bot$.
    \end{itemize}
\end{definition}

Unlike traditional approaches that rely on a single data structure, our method employs a \emph{Dual-Tree Architecture} for separating the roles of package manager and auditor, hence the properties of dependency inclusion and policy compliance.
The \emph{package tree} (\PT) is a \VC maintained by the package manager, and it is used for the construction of the inclusion proof \textbf{P1}, certifying that a specific dependency exists in the official repository without revealing it to external parties. 
The \emph{shadow tree} (\ST) is a \VC maintained by the auditor, and it is isomorphic to the \PT, sharing the exact same depth and indices. \ST is used for producing the proofs \textbf{P2} certifying that at a given leaf $j$, the package $\package_{j}$ within the PM is compliant to a set of policy constraints. 
The \textit{Dual-Tree Architecture} allows us to efficiently separate the policy evaluation from the proof generation. 
A way to proceed could have been to encode policy constraints directly within the arithmetic circuit (we used circuit in our implementation and we discuss this later in the implementation part).
However, this \textit{logic-in-circuit} paradigm might not have been scalable and maintainable. 
Every modification to a policy would necessitate re-compiling the circuit, re-distributing it and, in some schemes, re-running the trusted setup. Furthermore, the proof generation time and verification time would have increased significantly.

We utilize a specific variant of Merkle trees, \textit{sparse Merkle trees}~\cite{smtrees}, to force deterministic indexing for handling any positional shift of the packages that may invalidate the correct calculation of the Merkle paths.
For both the cases, to determine the location of a package \package within a tree, we use a simple formula like the following:

\begin{equation}
    i = \mathcal{H}(\text{Name}(\package_j)) \pmod{2^D}
\end{equation}

\noindent where $i$ denotes the sparse tree index and $j$ denotes the sequential package identifier, with $i \neq j$ in the general case, $D$ is the three depth, and $\mathcal{H}$ is the used hash function.
    
This creates a mapping where the index is derived solely from the package identity, and not from the arrival time on the public repository.
We hash the package name with a modular calculation where $D$ is the tree depth (e.g., $D=10$ or $20$). Regardless of how many other packages the package manager add to or remove from the public repository, this ensures that the position of a package is invariant, hence universally identifiable.
Instead, the value stored at each leaf is a cryptographic commitment of each package \textit{x} metadata for the \PT (Eq. \ref{eq:hash}), otherwise the flag value $V$ for the \ST indicating whether the package satisfies the user's security policies (Eq. \ref{eq:hash_compl}):
\begin{eqnarray}
    V_{x}[PT] = \mathcal{H}(\text{Metadata}(\package_x))
    \label{eq:hash}
\\    
    V_{x}[ST] = b \text{ where } b \in \{0, 1\}
    \label{eq:hash_compl}
\end{eqnarray}
where $1$ represents a compliant state and $0$ a non-compliant one.


\subsection{Proof Generation} 
The goal of dependency verification is to ensure that the list of dependencies included in an artifact’s SBOM is authentic and compliant with a set of policy constraints. 
For carrying out \textbf{P1} and \textbf{P2}, in \veribom\ we designed a zero-knowledge proof construction to generate verifiable claims in which the \prover\ demonstrates that all dependencies of a package are authentic (i.e., included in the \PT) and compliant with specified policies (i.e., included in the \ST), without revealing any sensitive metadata. 

\paragraph{Policy Constraints}
When a cryptographic proof related to a confidential SBOM is being generated, the authenticity of the dependencies and the compliance with the policy constraints should be strictly bounded. 
Therefore, if a verifier requires some policy constraints to be respected, although the software dependency is authentic but not compliant the verification of the proof must fail.
For accomplishing this, \veribom\ leverages a \textit{Dual-Membership Constraint}, i.e., for any verification step $k$, the \prover\ must demonstrate validity in both \PT and \ST.

More formally, given $\digest_{PT}$ as the public root of \PT and $\digest_{ST}$ be the public root of a \ST, for a package $\package$ at index $idx$, we define the validity predicate as follows:

\begin{align}\label{eq:zkp-dual}
    & \mathsf{ValidDep}(\package, \digest_{PT}, \digest_{ST}) = \nonumber \\
    & \quad \exists\ path_{PT}, path_{ST}\ \text{such that} \nonumber \\
    & \quad \verify(\digest_{PT}, \hashFun(\package), idx, path_{PT}) = \True \ \wedge \\
    & \quad \verify(\digest_{ST}, \mathbf{1}, idx, path_{ST}) = \True
\end{align}

\noindent where:
\begin{itemize}
    \item $path_{PT}$ is the Merkle Path that links a package to the \PT root $\digest_{PT}$.
    \item $path_{ST}$ is the Merkle Path that links a compliance value $\mathbf{1}$ to the \ST root $\digest_{ST}$ at index $idx$.
\end{itemize}

If the leaf in \ST is flagged as $0$ (Non-Compliant), the hashing sequence initiated with the forced value $1$ will diverge, and the root of the proof will result different from the public $\digest_{ST}$, causing the constraint to fail.

Since an SBOM contains $N$ packages, \veribom\ leverages the cryptographic folding scheme to aggregate multiple checks efficiently. 
The predicate $\mathsf{ValidDep}$ constitutes a single atomic step applied iteratively. 
At each step $k$, the folding scheme aggregates the validity of the step $k$ into a preexisting \textit{accumulator} $W_{k-1}$ (which holds the validity of the previous step $k-1$), producing a new accumulator $W_k$. 
The final proof $\pi$ attests that $\forall k \in [0, N-1], \mathsf{ValidDep}(\dots)$ holds true. 
\textbf{P1} and \textbf{P2} are ensured with constant-size verification cost.

\paragraph{Folding for Scalable Proof Generation}
In the proof construction, the \prover\ generates an opening for each dependency’s inclusion in the \VC and proves that it satisfies the policy constraints. These constraints are automatically enforced within the ZKP proof logic, removing the need for manual policy checks during verification. As a result, the final proof attests that (i) all dependencies are authentic, (ii) the policy constraints are satisfied, and (iii) the proof corresponds to the specific list of dependencies for the artifact. This enables privacy-preserving dependency validation, and enables downstream verifiers to check compliance without learning the specific dependencies involved and their metadata. 

To efficiently generate a proof and verify it for an arbitrary number of dependencies in practice, we employ folding-based zkSNARKs to address scalability issues related to proof size. As software ecosystems grow, the number of dependencies in a project can increase, reaching hundreds or even thousands. Thus, generating and verifying separate proofs for each dependency would be very expensive. Using a folding-based scheme, the \prover\ iteratively ``folds'' multiple proofs into a single compact proof, ensuring that the same verification rules apply to all dependencies while reducing proof size. 
An illustration on real world packages of the folding mechanism with the dual-membership constraint is provided in Figure~\ref{fig:trees_folding}.
In each proof step, corresponding to a specific dependency, this dependency is first verified against the public roots. 
Then, the result for this step is obtained by concatenating with the results of the previous step, and in turn is used as input to the next step. 
This process compresses multiple checks into a compact, succinct representation.
The verifier only needs to validate the final folded proof, which contains the verification results for all dependencies.

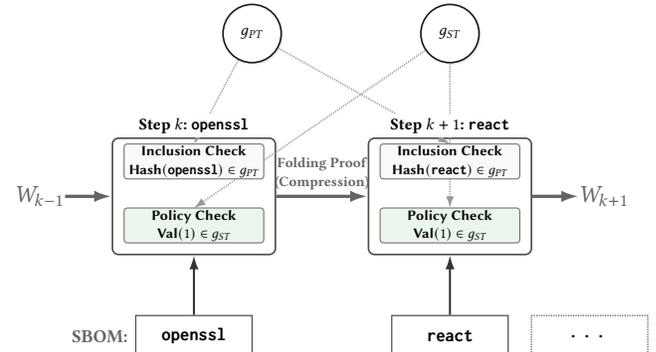
\begin{figure}[htbp]
    \centering
    \resizebox{\columnwidth}{!}{%
        \input{Figure/folding}
    }
    \caption{When at step \textit{k} a new package is retrieved from the SBOM, the proof is generated from the results of the previous step \textit{k-1}. The results of step \textit{k} are then employed for the step \textit{k+1}, hence for the next package. In each step, the proof is built considering the public roots.}
    \label{fig:trees_folding}
\end{figure}

\subsection{Policy Propagation Engine}
In a package repository, the packages may contain a great number of dependencies; manually and consistently applying the policy constraints can result very impractical.
Considering this scenario, to support \textbf{P2}, we designed a \textit{policy propagation engine} (\PPE) that for each package processes its graph of dependencies.
The \PPE is a layer for the activity of the Auditor that automatically flags the compliance status of dependent packages. 

\begin{definition}
Let $\mathcal{D}(p)$ be the set of the dependencies $d$ for a package $p$. 
The \textit{\ST compliance status} $\mathcal{C}(p) \in \{0, 1\}$ is the product of a package's compliance status with respect to a policy constraint $\mathcal{L}(p) \in \{0, 1\}$ and the compliance of all its dependencies $\mathcal{C}(d)$:

\begin{equation}\label{eq:propagation}
    \mathcal{C}(p) = \mathcal{L}(p) \cdot \prod_{d \in \mathcal{D}(p)} \mathcal{C}(d)
\end{equation}
\end{definition}

If a single dependency $d$ in the graph of dependencies has $\mathcal{C}(d)=0$, the zero propagates to the upper hierarchy, making $\mathcal{C}(p)=0$, thanks to the product operator $\prod$.
First the Auditor manually selects specific packages that are non-compliant with the established policy constraints; then, the \PPE traverses the dependency graph and upwards to all parent packages propagates the non-compliance status. 


\paragraph{Multiple Policy Constraints}
A user may want to check single and specific policy constraints rather than their combination (e.g., a user wants to prove only the absence of vulnerabilities, or the adhesion to a specific license, or both).
A key feature of \textsc{VeriBOM} is its ability to support heterogeneous policy constraints simultaneously. 
Since a leaf contains only a value 1/0, a single \ST is insufficient in handling multiple policy constraints. 
Within \veribom, the Auditor can maintain a \textit{forest} of $\ST$s, where each tree corresponds to a specific policy $\Pi$ (e.g., $\Pi_{Sec}$ for vulnerabilities, $\Pi_{Lic}$ for licenses).

\begin{figure}[htbp]
    \centering
    \resizebox{\columnwidth}{!}{%
        \input{Figure/shadow_trees}
    }
    \caption{Hierarchical propagation of policy constraints. The diagram illustrates how a local policy violation ($\mathcal{L}=0$) in a dependency ($P_A$) within a specific policy constraint propagates among other constraints, resulting in a non-compliant root package ($P_R$).}
    \label{fig:propagation}
\end{figure}
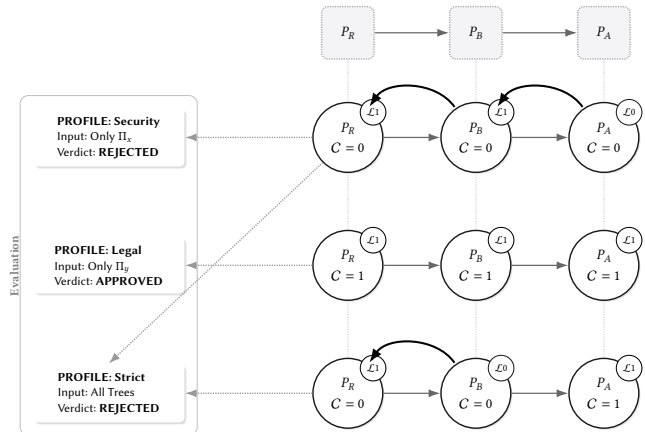

While the PM Tree remains immutable, the \PPE generates distinct $\ST$s compliance states for each profile:
\begin{equation}
    \mathcal{C}(p, \Pi) = \mathcal{L}(p, \Pi) \cdot \prod_{d \in \mathcal{D}(p)} \mathcal{C}(d, \Pi)
\end{equation}

Furthermore, to support complex requirements i.e., ensuring a package is compliant to combination of policy constraints, the \PPE allows for \textit{policy composition}. 
Let $\mathbf{\Psi} = \{\Pi_1, \dots, \Pi_n\}$ be a set of policy constraints.
The compliance state for this composition is derived by computing the intersection of all individual compliance states:

\begin{equation}\label{eq:composition}
    \mathcal{C}(p, \mathbf{\Psi}) = \prod_{\Pi_k \in \mathbf{\Psi}} \mathcal{C}(p, \Pi_k)
\end{equation}

Figure \ref{fig:propagation} describes the propagation mechanism along with a set of policy constraints.

During the Auditing phase, the Auditor publishes a distinct root $\digest_{ST}^{(k)}$ for each $\Pi_k$.
Instead, in the verification phase, the Client chooses $\Pi_k$ and retrieves the corresponding public root $\digest_{ST}^{(k)}$, against which the ZK proof is then verified.

\paragraph{Client-on-demand Aggregation of Policy Constraints}
A user may require specific policy constraints that have not been covered by the auditor. Within \veribom\, we consider this dynamic case as \textit{client-on-demand aggregation}.
In the default setting, the auditor proactively maintains $\Psi$ that are of common requests. 
The client specifies a composite requirement $\mathbf{\Psi}_{\text{Target}} = \Psi_{\text{Auditor}} \cup \Psi_{\text{Client}}$.
A transient \ST root $\digest_{ST}^{(\text{Target})}$ is generated by the union of the policy constrains. A package $p$ is compliant if and only if it satisfies the entirety of the policy constraints $\mathcal{C}(p, \mathbf{\Psi}_{\text{Target}})$: 
    \begin{equation}
    \mathcal{C}(p, \mathbf{\Psi}_{\text{Target}}) = 
    \underbrace{\left( \prod_{\Pi \in \mathbf{\Psi}_{\text{Auditor}}}\mathcal{C}(p, \Pi) \right)}_{\text{Auditor's}} 
    \cdot 
    \underbrace{\left( \prod_{\Pi \in \mathbf{\Psi}_{\text{Client}}} \mathcal{C}(p, \Pi) \right)}_{\text{Client's}}
\end{equation}

To ensure unambiguous verification, we define the compliance with policy constraints using Datalog semantics.
$\mathcal{C}(p, \mathbf{\Psi}_{\text{Target}})$ is calculated as logical conjuction of the singular policy constraints.

Let $\mathbf{\Psi}_{\text{Target}} = \{\Pi_1, \Pi_2, \dots, \Pi_n\}$ be the set of the policy constraints. We express the aggregation rule as:

\begin{equation}
    \mathcal{C}(p, \mathbf{\Psi}_{\text{Target}}) \leftarrow \Pi_1(p) \wedge \Pi_2(p) \wedge \dots \wedge \Pi_n(p)
\end{equation}

where each literal $\Pi_k(p)$ evaluates to true if and only if package $p$ is flagged as compliant ($1$) in the specific Shadow Tree corresponding to policy constraint $\Pi_k$.

\paragraph{Accountability via Root Binding}
As described in the threat model part, a software vendor may act as a malicious vendor ($\mathcal{A}_V$), i.e., a vendor that in an SBOM adds some packages not authentic nor compliant.
It is crucial to note that the ZK-SNARK soundness does not prevent $\mathcal{A}_V$ from building a proof that is based on different \PT and \ST. $\mathcal{A}_V$ could theoretically clone the structure of the package repository, build its own \textbf{$PT^*$} and \textbf{$ST^*$}, and derive divergent roots $\digest'_{PT^*}$ and $\digest'_{ST^*}$.
Consequently, the Vendor can successfully generate a Zero-Knowledge proof $\pi_{*}$ that is mathematically valid with respect to these divergent roots.

However, the security of \textsc{VeriBOM} relies on the binding of the verification to the trusted anchors.
In Phase 3, the Client ignores $\digest'_{PT^*}$ and $\digest'_{ST^*}$ provided by the software vendor and retrieves the public roots $\digest_{PT}$ and $\digest_{ST}$ directly from the trusted entities.
The verification fails due to roots mismatch:
\begin{equation}
    \verify(\pi_{*}(\digest'_{PT^*},\digest'_{ST^*}), \digest_{PT}, \digest_{ST}) = \False
\end{equation}
Since the proof $\pi_{*}$ is bound to the inputs used to generate it ($\digest'_{*}$), the verification phase of \veribom\ rejects the proof.
This forces strong accountability: the vendor cannot create a parallel-but-correct reality as the generated proof does not align with the public roots.

\subsection{Dependencies Verification}
When a consumer or any downstream user plans to use software, they can verify the authenticity of its dependencies to ensure no compromise and compliance with declared policies without exposing sensitive information, such as secret dependencies, by using the zero-knowledge proofs generated by \veribom. These proofs are used to verify conformance to security and privacy policies defined by stakeholders, and can be independently verified by anyone, e.g., consumers or regulators. 
The public SBOM may include public data such as the software identifier, the hash of the dependency list, license information, and allow-lists. This information in addition to the proof, which is a hash value, is used as input to the $verify\_zkp\_proof$ algorithm to check the truthfulness of the corresponding statement.

\subsection{Operational Workflow}
The operational lifecycle of \veribom\ orchestrates the interaction between trusted infrastructure (Package Manager and Auditor) and supply chain actors (Vendor and Client) through three sequential phases, as formalized in Algorithm~\ref{alg:zksbom_protocol}.
Figure \ref{fig:sbom_obfuscation} clarifies the transition from a full public SBOM to an SBOM with confidential information.

\paragraph{Setup Phase.}
\veribom\ begins with the setup the global state. 
The package manager populates the \textit{Package Tree} ($PT$) with commitments to each package, producing a vector commitment.
The \textbf{Auditor} based on instantiated specific policy constraints generates an isomorphic \textit{Shadow Tree} ($ST_{\mathbf{\Psi}}$) where each leaf contains the propagated compliance status of the corresponding package.
This setup phase is concluded with the publication of the roots $\digest_{PT}$, which handles \textbf{P1}, and $\digest_{ST}^{\mathbf{\Psi}}$, which handles \textbf{P2}.

\paragraph{Prove Phase.}
In this phase, the software vendor, or the \prover\ (the software vendor can act as a \prover\ too), generates a cryptographic proof for protecting the confidentiality of their SBOM. 
The \prover\ iterates through the software's dependency graph described in the SBOM, selects the dependencies to be hidden, and generates a Zero-Knowledge proof locally.
Using a recursive folding scheme, a succinct proof $\pi$ is produced. 
Therefore, the software developer preserves the confidentiality of its SBOM, by publishing only the proof $\pi$ attesting that the SBOM satisfies \textbf{P1} and \textbf{P2}.

\paragraph{Verify Phase}
In the last phase, the client (or code consumer) validates the software artifact by not dissecting an SBOM but instead verifying that the received proof $\pi$ is valid against the public roots $\digest_{PT}$ and $\digest_{ST}^{\mathbf{\Psi}}$, selecting the policy constraints to be verified.
If $\pi$ holds, the code consumer obtains a mathematical guarantee that the SBOM of the vendor, hence the hidden dependency graph, is fully consistent and compliant with respect to the selected policy constraints.
The validity of the proof authenticates the software artifact without revealing any confidential information.

\begin{algorithm}[htbp]
\SetKwInput{KwData}{Input}
\SetKwInput{KwResult}{Output}
\SetKwFunction{Hash}{Hash}
\SetKwFunction{MerkleRoot}{MerkleRoot}
\SetKwFunction{MerklePath}{MerklePath}
\SetKwFunction{VerifyProof}{Verify}
\SetKwFunction{Fold}{Fold}
\SetKwFunction{Compress}{Compress}
\SetKw{True}{true}
\SetKw{False}{false}

\caption{ZkSBOM Supply Chain Verification Protocol}
\label{alg:zksbom_protocol}

\tcc{--- Setup Phase ---}
\KwData{Repository $\mathcal{R} = \{P_1, \dots, P_N\}$, Policy Set $\mathbf{\Psi}$}
\KwResult{Public Roots $\digest_{PT}, \digest_{ST}^{\mathbf{\Psi}}$}

\tcc{Step 1.1: Package Manager (Identity)}
Initialize tree $PT$\;
\For{$i \leftarrow 0$ \KwTo $N-1$}{
    $P_i \leftarrow \mathcal{R}[i]$\;
    $h_i \leftarrow \Hash(P_i.\text{name}, P_i.\text{ver}, P_i.\text{lic})$ \tcp*{Identity Commitment}
    Insert leaf $h_i$ into $PT$ at index $i$\;
}
$\digest_{PT} \leftarrow \MerkleRoot(PT)$\;
Publish $\digest_{PT}$\;

\BlankLine
\tcc{Step 1.2: Auditor (via \PPE)}
Initialize Shadow Tree $ST_{\mathbf{\Psi}}$\;
\For{$i \leftarrow 0$ \KwTo $N-1$}{
    \tcp{Calc recursive status (Eq. \ref{eq:propagation})}
    $s_i \leftarrow \prod_{k \in \mathbf{\Psi}} \mathcal{C}(P_i, \Pi_k)$\;
    Insert leaf $s_i$ into $ST_{\mathbf{\Psi}}$ at index $i$\;
}
$\digest_{ST}^{\mathbf{\Psi}} \leftarrow \MerkleRoot(ST_{\mathbf{\Psi}})$\;
Publish $\digest_{ST}^{\mathbf{\Psi}}$\;

\BlankLine
\tcc{--- Prove Phase: \prover\ (Vendor) ---}
\KwData{Dependencies $\mathcal{D}=\{d_1 \dots d_K\}$, Trees $PT, ST_{\mathbf{\Psi}}$}
\KwResult{Succinct Proof $\pi$}

Initialize recursive proof state $\pi_0$\;
\For{$j \leftarrow 1$ \KwTo $K$}{
    Identify index $idx$ of $d_j$ in $PT$\;
    $path_{PT} \leftarrow \MerklePath(PT, idx)$\;
    $path_{ST} \leftarrow \MerklePath(ST_{\mathbf{\Psi}}, idx)$\;
    
    \tcc{Fold step checking P1 and P2}
    $\pi_j \leftarrow \Fold(\pi_{j-1}, d_j, path_{PT}, path_{ST}, \digest_{PT}, \digest_{ST}^{\mathbf{\Psi}})$\;
}
\Return $\pi \leftarrow \Compress(\pi_K)$\;

\BlankLine
\tcc{--- Verify Phase (Client) ---}
\KwData{Proof $\pi$, Trusted Roots $\digest'_{PT}, \digest'^{\mathbf{\Psi}}_{ST}$}
\KwResult{Accept or Reject}

\tcp{Verification binds proof to public roots}
$Valid \leftarrow \VerifyProof(\pi, \digest'_{PT}, \digest'^{\mathbf{\Psi}}_{ST})$\;
\eIf{$Valid$ is \True}{
    \Return \textbf{Accept} \tcp*{Authentic \& Compliant}
}{
    \Return \textbf{Reject} \tcp*{Root Mismatch or Policy Fail}
}
\end{algorithm}

\begin{figure}[t]
    \centering
    \begin{minipage}[t]{0.45\columnwidth}
        \begin{lstlisting}[language=json, basicstyle=\ttfamily\scriptsize, frame=single, numbers=none, caption={Internal SBOM}, label={lst:sbom.private}]
{
  "name": "banking-core",
  "version": "2.4.1",
  "dependencies": [
    {
      "name": "tokio",
      "version": "1.28.0",
      "src": "crates.io"
    },
    {
      "name": "log4rs",
      "version": "1.2.0",
      "src": "crates.io"
    }
  ]
}
        \end{lstlisting}
    \end{minipage}
    \hfill
    \begin{minipage}[t]{0.45\columnwidth}
        \begin{lstlisting}[language=json, basicstyle=\ttfamily\scriptsize, frame=single, numbers=none, caption={Public SBOM}, label={lst:sbom.public}]
{
  "name": "banking-core",
  "version": "2.4.1",
  "dependencies": [
    {
      "name": "tokio",
      "version": "1.28.0",
      "src": "crates.io"
    },
    {
      "commitment": "0x7f8a...",
      "type": "zk-hidden"
    }
  ],
  "zk_proof": "pi_snark_0x99a1..."
}
        \end{lstlisting}
    \end{minipage}
    
    \caption{Comparison of SBOM visibility: software vendor view (left) vs. verifier view (right) after \veribom application.}
    \label{fig:sbom_obfuscation}
\end{figure}

%% file: Sec-ThreatModel.tex
\subsection{Threat Model}
\label{sec:tm}

\paragraph{Confidentiality Requirements of SBOMs}
Traditional SBOMs mandate full transparency, however we identify the following critical cases requiring confidentiality:

\begin{itemize}
    \item \emph{Intellectual property protection} In software development, distinct delivery models exist. In proprietary software, the specific composition of internal dependencies is to be considered confidential. In such a case, the consumer receives a self-contained artifact where transitive dependencies are implicitly trusted \cite{10.5555/3361338.3361407}. The same goes for vendors that curate databases of non-public vulnerabilities (e.g., zero-days). 
    
    \item \emph{Components unpatching} Developers may include outdated or vulnerable elements in the supply chain that cannot be patched and that, for some technological and structural reason, are continuously being used by the supply chain \cite{10.5555/3361338.3361435}.
    Also, empirical studies in software engineering have found that update propagation is slow. Some research \cite{10.1007/s10664-017-9521-5} indicates that up to 69\% of developers are unaware of their vulnerable dependencies, and developers are less likely to migrate their library dependencies, with up to 81.5\% of systems keeping outdated dependencies because many developers fear that library updates may introduce breaking changes.
    
    \item \emph{Business distrust} In multi-tier supplychains (where  a network of software suppliers exists),  Tier-1 suppliers generally integrate components from specialized sub-suppliers (Tier-2+). 
    A significant commercial risk is that a transparent SBOM may lead to the disintermediation of Tier-1 suppliers, as customers could attempt to bypass them and engage directly with sub-suppliers.
\end{itemize}

\paragraph{Attackers in SBOMs}
Based on the operational scenarios identified above, we abstract the following classes of attackers with different capabilities: 

\begin{itemize}
    \item The \emph{targeting attacker (\tattacker)} acts as a malicious observer by performing \textit{fingerprinting} and learning the dependencies and/or vulnerabilities within the plaintext SBOM. By analyzing the SBOM, \tattacker may reduce the search space for targeting known exploits. 

    \item The \emph{competitor attacker (\cattacker)} acts as a business competitor or a supply chain partner aiming for commercial gain. By analyzing the SBOM, \cattacker may reconstruct the business graph to bypass intermediaries.

    \item The \emph{malicious vendor (\mattacker)}, mainly internal at the software processes. \mattacker acts as a malicious prover by claiming false properties about the SBOM, such as vulnerable or outdated packages, or by not complying with the security policies requested for the final approval of the delivered software.
\end{itemize}

%% file: Figure/folding.tex
\pgfdeclarelayer{bg}
\pgfsetlayers{bg,main}

\begin{tikzpicture}[
    every node/.style={font=\bfseries\fontsize{40}{50}\selectfont},
    root node/.style={
        circle, draw=black!90, line width=4pt, minimum size=5cm, 
        align=center, fill=white,
        drop shadow={opacity=0.15, shadow xshift=5pt, shadow yshift=-5pt}
    },
    step box/.style={
        rectangle, rounded corners=15pt, draw=black!70, fill=white,
        line width=4pt, 
        minimum width=14cm, minimum height=10cm, 
        align=center
    },
    verify block/.style={
        rectangle, draw=black!60, fill=gray!5,
        minimum width=12cm, minimum height=3cm, 
        font=\sffamily\bfseries\fontsize{35}{45}\selectfont, 
        align=center, rounded corners=8pt,
        line width=2pt
    },
    tape item/.style={
        rectangle, draw=black!70, fill=white,
        minimum width=10cm, minimum height=3.5cm, 
        font=\ttfamily\bfseries\fontsize{45}{55}\selectfont, 
        line width=3pt
    },
    input arrow/.style={-{Latex[length=12mm, width=8mm]}, line width=5pt, color=black!80},
    acc arrow/.style={-{Latex[length=15mm, width=10mm]}, line width=8pt, color=accColor},
    const arrow/.style={-{Latex[length=10mm, width=7mm]}, line width=4pt, dashed, color=black!40}
]

\node[step box] (Step1) at (0, 0) {};
\node[step box] (Step2) at (22, 0) {}; 

\node[root node] (RootPT) at (5, 14) {$\digest_{PT}$};
\node[root node] (RootST) at (22, 14) {$\digest_{ST}$};

\draw[const arrow] (RootPT) -- (0, 3 + 1.5); 
\draw[const arrow] (RootST) -- (0, -2.5 + 1.5);

\draw[const arrow] (RootPT) -- (22, 3 + 1.5); 
\draw[const arrow] (RootST) -- (22, -2.5 + 1.5);

\node[above, fill=white, inner sep=10pt] at (Step1.north) {Step $k$: \texttt{openssl}};
\node[above, fill=white, inner sep=10pt] at (Step2.north) {Step $k+1$: \texttt{react}};

\node[verify block] (VerPT_1) at (0, 3) {Inclusion Check\\$\text{Hash}(\texttt{openssl}) \in \digest_{PT}$};
\node[verify block, fill=stGreen!10] (VerST_1) at (0, -2.5) {Policy Check\\$\text{Val}(1) \in \digest_{ST}$};

\node[verify block] (VerPT_2) at (22, 3) {Inclusion Check\\$\text{Hash}(\texttt{react}) \in \digest_{PT}$};
\node[verify block, fill=stGreen!10] (VerST_2) at (22, -2.5) {Policy Check\\$\text{Val}(1) \in \digest_{ST}$};

\node[tape item] (Tape1) at (0, -12) {\texttt{openssl}};
\node[tape item] (Tape2) at (22, -12) {\texttt{react}};
\node[tape item, dashed] (Tape3) at (34, -12) {\dots};

\node[left=0.5cm of Tape1, color=black!60, font=\bfseries\fontsize{45}{55}\selectfont] {SBOM:};

\draw[input arrow] (Tape1) -- (Step1.south);
\draw[input arrow] (Tape2) -- (Step2.south);

\node[left=4cm of Step1, color=accColor, font=\bfseries\fontsize{60}{70}\selectfont] (AccIn) {$W_{k-1}$};
\draw[acc arrow] (AccIn) -- (Step1);

\draw[acc arrow] (Step1) -- node[midway, above, color=accColor, align=center, yshift=10pt, font=\bfseries\fontsize{35}{45}\selectfont] 
    {Folding Proof\\(Compression)} (Step2);

\node[right=4cm of Step2, color=accColor, font=\bfseries\fontsize{60}{70}\selectfont] (AccOut) {$W_{k+1}$};
\draw[acc arrow] (Step2) -- (AccOut);

\end{tikzpicture}

%% file: Figure/shadow_trees.tex
\pgfdeclarelayer{bg}
\pgfsetlayers{bg,main}

\begin{tikzpicture}[
    every node/.style={font=\bfseries\fontsize{35}{45}\selectfont},
    shadow node/.style={
        circle,
        draw=black!90,
        line width=3.5pt,
        minimum size=6cm,
        align=center,
        fill=white,
        drop shadow={opacity=0.2, shadow xshift=4pt, shadow yshift=-4pt}
    },
    physical node/.style={
        rectangle, rounded corners=15pt, draw=physGray!80, fill=physGray!10,
        line width=3.5pt, 
        minimum size=4.5cm,
        align=center, 
        dashed
    },
    badge/.style={
        circle, draw=black!90, line width=2.5pt, 
        minimum size=2.2cm,
        font=\fontsize{25}{30}\selectfont\bfseries,
        fill=white, inner sep=0pt,
        anchor=center,
        drop shadow={opacity=0.15, shadow xshift=2pt, shadow yshift=-2pt}
    },
    verdict node/.style={
        rectangle, rounded corners=12pt, line width=3pt,
        minimum width=13cm, minimum height=4.5cm, 
        align=left, 
        font=\sffamily\fontsize{30}{38}\selectfont,
        drop shadow={opacity=0.25, shadow xshift=5pt, shadow yshift=-5pt},
        inner sep=20pt, anchor=east,
        fill=white
    },
    badge compliant/.style={text=black, draw=black},
    badge noncompliant/.style={text=black},
    compliant/.style={text=black},
    noncompliant/.style={},
    dep arrow/.style={-{Latex[length=10mm, width=7mm]}, line width=3pt, color=black!60},
    prop arrow/.style={-{Latex[length=12mm, width=9mm]}, line width=5pt, shorten <=4pt, shorten >=4pt},
    map arrow/.style={-, line width=3pt, color=physGray!40, dashed},
    result arrow/.style={-{Latex[length=10mm, width=7mm]}, line width=3.5pt, color=black!40, dashed}
]

\def\xRes{-14}   
\def\xRoot{0}    
\def\xMid{11}    
\def\xDeep{22}   

\def\yPhys{9}    
\def\ySec{0}     
\def\yLic{-11}   
\def\yDep{-22}   

\node[physical node] (P_Root) at (\xRoot, \yPhys) {$P_R$};
\node[physical node] (P_Mid)  at (\xMid, \yPhys)  {$P_B$};
\node[physical node] (P_Deep) at (\xDeep, \yPhys) {$P_A$};
\node[anchor=south west, color=physGray!80, font=\bfseries\fontsize{40}{50}\selectfont] at (-6, \yPhys+2) {};

\node[shadow node, noncompliant] (S_Root) at (\xRoot, \ySec) {$P_R$\\[5pt]$\mathcal{C}=0$};
\node[badge, badge compliant] at (S_Root.45) {$\mathcal{L}1$};
\node[shadow node, noncompliant] (S_Mid) at (\xMid, \ySec) {$P_B$\\[5pt]$\mathcal{C}=0$};
\node[badge, badge compliant] at (S_Mid.45) {$\mathcal{L}1$};
\node[shadow node, noncompliant] (S_Deep) at (\xDeep, \ySec) {$P_A$\\[5pt]$\mathcal{C}=0$};
\node[badge, badge noncompliant] at (S_Deep.45) {$\mathcal{L}0$};

\node[shadow node, compliant] (L_Root) at (\xRoot, \yLic) {$P_R$\\[5pt]$\mathcal{C}=1$};
\node[badge, badge compliant] at (L_Root.45)  {$\mathcal{L}1$};
\node[shadow node, compliant] (L_Mid) at (\xMid, \yLic) {$P_B$\\[5pt]$\mathcal{C}=1$};
\node[badge, badge compliant] at (L_Mid.45) {$\mathcal{L}1$};
\node[shadow node, compliant] (L_Deep) at (\xDeep, \yLic) {$P_A$\\[5pt]$\mathcal{C}=1$};
\node[badge, badge compliant] at (L_Deep.45) {$\mathcal{L}1$};

\node[shadow node, noncompliant] (D_Root) at (\xRoot, \yDep) {$P_R$\\[5pt]$\mathcal{C}=0$};
\node[badge, badge compliant] at (D_Root.45){$\mathcal{L}1$};
\node[shadow node, noncompliant] (D_Mid) at (\xMid, \yDep) {$P_B$\\[5pt]$\mathcal{C}=0$};
\node[badge, badge noncompliant] at (D_Mid.45)  {$\mathcal{L}0$};
\node[shadow node, compliant] (D_Deep) at (\xDeep, \yDep) {$P_A$\\[5pt]$\mathcal{C}=1$};
\node[badge, badge compliant] at (D_Deep.45){$\mathcal{L}1$};

\node[verdict node] (Res_Sec) at (\xRes, \ySec) {
    \textbf{PROFILE: Security}\\
    Input: Only $\Pi_{x}$\\
    Verdict: \textbf{REJECTED}
};
\node[verdict node] (Res_Lic) at (\xRes, \yLic) {
    \textbf{PROFILE: Legal}\\
    Input: Only $\Pi_{y}$\\
    Verdict: \textbf{APPROVED}
};
\node[verdict node] (Res_Strict) at (\xRes, \yDep) {
    \textbf{PROFILE: Strict}\\
    Input: All Trees\\
    Verdict: \textbf{REJECTED}
};

\draw[dep arrow] (P_Root) -- (P_Mid);
\draw[dep arrow] (P_Mid)  -- (P_Deep);
\draw[dep arrow] (S_Root) -- (S_Mid);
\draw[dep arrow] (S_Mid) -- (S_Deep);
\draw[dep arrow] (L_Root) -- (L_Mid);
\draw[dep arrow] (L_Mid) -- (L_Deep);
\draw[dep arrow] (D_Root) -- (D_Mid);
\draw[dep arrow] (D_Mid) -- (D_Deep);

\draw[prop arrow] (S_Deep) to[bend right=55] (S_Mid);
\draw[prop arrow] (S_Mid) to[bend right=55] (S_Root);
\draw[prop arrow] (D_Mid) to[bend right=55] (D_Root);

\draw[result arrow] (S_Root) -- (Res_Sec);
\draw[result arrow] (L_Root) -- (Res_Lic);
\draw[result arrow] (D_Root) -- (Res_Strict);
\draw[result arrow] (S_Root.south west) --(Res_Strict.north);

\begin{pgfonlayer}{bg}
    \draw[map arrow] (P_Root.center) -- (S_Root.center);
    \draw[map arrow] (S_Root.center) -- (L_Root.center);
    \draw[map arrow] (L_Root.center) -- (D_Root.center);
    \draw[map arrow] (P_Mid.center) -- (S_Mid.center);
    \draw[map arrow] (S_Mid.center) -- (L_Mid.center);
    \draw[map arrow] (L_Mid.center) -- (D_Mid.center);
    \draw[map arrow] (P_Deep.center) -- (S_Deep.center);
    \draw[map arrow] (S_Deep.center) -- (L_Deep.center);
    \draw[map arrow] (L_Deep.center) -- (D_Deep.center);
    
    \node[draw=gray!50, rounded corners=20pt, fit=(Res_Sec) (Res_Lic) (Res_Strict), inner sep=30pt] (DashBox) {};
    \node[rotate=90, anchor=south, color=black!50, font=\bfseries\fontsize{30}{40}\selectfont] at (DashBox.west) {Evaluation};
\end{pgfonlayer}

\end{tikzpicture}

%% file: Sec4-Soundness.tex
\section{Security Analysis}
\label{sec:security-analysis}
In this section, we analyze the properties that \veribom\ guarantees against a generic PPT adversary \attacker.
We consider a {probabilistic polynomial-time (PPT)} adversary $\mathcal{A}$, i.e., an algorithm running in polynomial time with respect to the security parameter and capable of making random choices. 
This means that $\mathcal{A}$ may use {randomized strategies}, i.e., using randomness during its execution to increase its chances of success. The adversary is assumed to have the capability to observe, intercept, and modify any communication between system entities. It may inspect the SBOM, including partial disclosures, and attempt to forge proofs or obscure the presence of malicious dependencies. 
{Furthermore, we assume that the adversary may compromise software components, including those within our proposed framework, therefore, gaining full control over their execution.

\paragraph{Integrity of the Package Tree}
We must ensure that in the package tree ($PT$), \attacker cannot forge a package once committed at index $idx$.

\begin{theorem}[Package Tree Position-Binding]\label{thm:pack.repo.integrity}
    \attacker cannot open a leaf at index $idx$ to a value $v'$ different from $v$, as the originally committed value.
\end{theorem}

\noindent\emph{Attack Implication:} 
\attacker could commit to a benign package to satisfy initial checks, but later successfully prove the inclusion of a malicious payload at the exact same identifier, effectively injecting malware into the supply chain.

\begin{proof}[Proof Sketch]
To be able to break the position-binding property, \attacker must successfully output two valid openings $\open(idx, v, \digest_{PT})$ and $\open(idx, v', \digest_{PT})$ with $v \neq v'$.
Hence, \attacker has found a collision in the underlying hash function $\hashFun$ or a collision in the structured hash $\structuredHashFunc$. 
Since we assume $\hashFun$ (e.g., Poseidon) to be collision-resistant, we consider the probability of such a case negligible.
\end{proof}

\paragraph{Soundness of the Dual-Tree Verification}
We must ensure that \textbf{P1} and \textbf{P2} are mutually satisfied.

\begin{theorem}[Soundness of Dual-Tree Verification]\label{thm:zkp.soundness}
    \attacker cannot produce a valid proof $\pi$ for a non-compliant or non-existent dependency package $\package$.
\end{theorem}

\noindent \emph{Attack Implication}
\attacker could include packages not existing within the package repository and include non-compliant packages, while generating a valid proof that falsely certifies both the contrary.

\begin{proof}[Proof Sketch]
    \attacker must succeed in at least one of the following forgeries:
    
    \begin{enumerate}
        \item \emph{Inclusion Forgery:} \attacker proves inclusion in $PT$ for a package that does not exist in the package repository. This requires forging a Merkle Path against $\digest_{PT}$, violating Theorem~\ref{thm:pack.repo.integrity}.
        
        \item \emph{Compliance Forgery:} \attacker proves compliance for a non-compliant package. Since the index $idx$ is deterministically derived from the package identity, proving compliance against $\digest_{ST}$ for a non-compliant status requires finding a hash collision to assert that a different path leads to the same root $\digest_{ST}$.
    \end{enumerate}

\end{proof}

\paragraph{Security against Parallel Tree Attacks (Accountability)}
We must ensure that \attacker cannot provide valid but forged proofs.

\begin{theorem}[Root Binding Accountability]
    \attacker cannot convince a code consumer to accept a valid proof generated for divergent \textbf{$\PT$*} and \textbf{$\ST$*}.
\end{theorem}

\noindent\emph{Attack Implication}
\attacker could maintain a private and forged state where malicious packages are authorized. \attacker could generate mathematically valid proofs for this parallel reality and bypass the public auditing process without detection.

\begin{proof}
    Let \attacker create parallel \textbf{$\PT$*} and \textbf{$\ST$*}, where respectively can exist malicious packages and contemporarily flagged as compliant with the policy constraints. 
    The \attacker can generate a mathematically valid proof $\pi*$ such that $\verify(\pi*, \digest'_{PT*}, \digest'_{ST*}) = \True$.
    
    However, \veribom is parameterized to rely on the trusted roots $\digest_{PT}$ and $\digest_{ST}$. 
    The verification check performed is \\ $\verify(\pi*, \digest_{PT}, \digest_{ST})$.
    Since the hash function is collision-resistant, $\digest'_{PT*} \neq \digest_{PT}$ and $\digest'_{ST*} \neq \digest_{ST}$.
    For the proof $\pi*$ to verify against $\digest_{PT}$ and $\digest_{ST}$, the adversary would essentially need to prove that the different pairs of trees have the same root, which is a hash collision. Thus, the verification returns $\False$ with non-negligible probability.
\end{proof}

The proposed architecture and the cryptographic primitives employed enable \veribom\ to address the aforementioned challenges effectively. However, as with any distributed system, \veribom\ must be evaluated against established adversarial models. We will discuss how the design of \veribom\ inherently mitigates the attacks from the attacker categories introduced in the threat model, i.e., the targeting attacker ($\mathcal{A}_T$), the competitive attacker ($\mathcal{A}_C$), and of the malicious vendor ($\mathcal{A}_V$).

\begin{itemize}[leftmargin=*]
    \item \emph{Preservation of confidentiality.}
    The targeting attacker $\mathcal{A}_T$ aims to fingerprint the software to identify possible  components that may contain specific vulnerable libraries.
    Since SBOM and the dependency graph with sensitive information are not explicitly shared to the client or made publicly available, $\mathcal{A}_T$ should perform {black-box analysis}, such as static/dynamic analysis to extract the dependencies.
    Considering the vast search space of existing libraries and version combinations, it becomes computationally impractical for the attacker to find the correct exploitable versions without knowing its dependency graph.

    \item \emph{Resistance to competitive analysis.}
    The competitive attacker $\mathcal{A}_C$ aims to reconstruct the graph of dependencies in order to exfiltrate any intermediate software supplier, hence to gain business advantage.
    Relying on the {Zero-Knowledge} property of the cryptographic scheme, the generated proof $\pi$ reveals no information about the witness $w$, which is the underlying SBOM that only the software vendor knows.
    The cryptographic proof is the only information used for checking the validity of the proof itself, not the data it verifies.
    Therefor, for $\mathcal{A}_C$, reconstructing the proprietary dependency tree, and hence the chain of supplier would require breaking the underlying cryptographic hardness assumptions.

    \item \emph{Prevention from version equivocation.}
    $\mathcal{A}_V$ generates a proof for a compliant artifact $A$ while delivering a malicious artifact $A'$. 
    In this case, it may result that the software has vulnerable dependencies while the proof ensures the contrary.
    \veribom\ mitigates this scenario via binding the artifact's digest in the proof itself. 
    There would result computationally infeasible to find a second witness $w'$ such that $Commit(w') = Commit(w)$. Thus, the proof is strictly bound to the specific delivered software version.

    \item \emph{Prevention from non-compliance.}
    $\mathcal{A}_V$ attempts to deceive the client by generating a valid proof $\pi$ for a software that violates the security policies established by the auditor and requested by the client itself.
    \textsc{VeriBOM} enforces integrity through a {dual-tree binding} mechanism: every leaf in the PT is mathematically mapped by the auditor in a \ST that aggregates the results of the auditing.
    If a software dependency is flagged as non-compliant (i.e., value $0$) in the \ST, and the software component is a dependency within the SBOM, $\mathcal{A}_V$ cannot produce a proof that results correct against $\digest_{ST}$.
    The cryptographic constraints would demonstrate that the generated proof $\pi$ not derivable from the $\digest_{PT}$ and $\digest_{ST}$.

    \item \emph{Resistance to replay attacks.}
    $\mathcal{A}_V$ deceives the verifier submitting a previously generated valid proof $\pi_{old}$. 
    In case the software is not updated, the prevention from version equivocation would not work.
    Instead, \veribom\ mitigates replay attacks by enforcing \textit{cryptographic freshness}: the validity of a proof $\pi$ is strictly bound to $\digest_{\PT}$ and $\digest_{\ST}$ that are publicly available and updated.
    Since any update to the package repository or the audit policies results in an update of the $\digest_{\PT}$ and $\digest_{\ST}$, a proof generated by embedding outdated roots will mathematically fail verification.
\end{itemize}

%% file: Sec5-Evaluation.tex
\section{Evaluation}
\label{sec:evaluation}
\subsection{Implementation}












\begin{figure}
	\centering
	\includegraphics[scale=0.8]{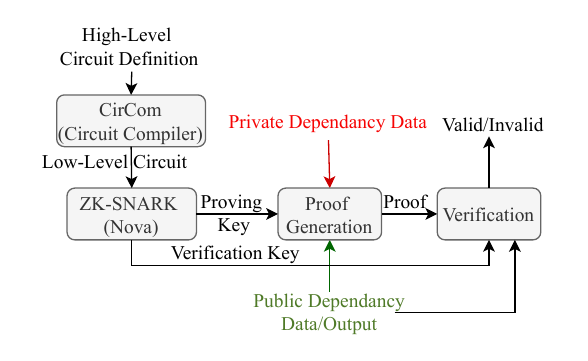}
	\caption{Implementation architecture of our ZK-Proof system.}
	\label{fig:zkp.system.implementation}
\end{figure}
 
The implementation of \veribom\ covers the entire workflow, from vector commitments construction for package repositories passing from proof generation to proof verification.  

\paragraph{Orchestration}
\textbf{Python} handles the general data management and the library \textbf{pyDataLog} handles the definition and progation of the Horn Clauses. A \textbf{SQLite} backend supports the auditor's compliance status (the ST). It maintains the state of the auditor's ST, allowing the Python Semantic Engine to efficiently query dependency graphs and propagate compliance flags before interacting with the cryptographic \prover.

\paragraph{Cryptographic Core}
The performance-critical cryptographic operations are executed in \textbf{Rust}.
\veribom\ automatically extracts and transforms package metadata from software repositories into a zero-knowledge-compatible format.  
To efficiently prove properties over variable-length SBOMs, we adopted cryptographic libraries for Merkle trees, vector commitments, and zkSNARKs using \textbf{Nova}~\cite{nova-paper, nova, nova-scotia}, 
Nova is a high-speed recursive folding scheme that realizes Incrementally Verifiable Computation (IVC). Unlike traditional SNARKs (e.g., Groth16), Nova allows us to verifies one package at a time. In the generation of the proof, the \prover, i.e., the software vendor, iterates through the SBOM list and accumulates the cryptographic proofs at a constant size. 
The final compressed proof $\pi$ is generated only at the end of the iteration, ensuring that the proof size remains succinct regardless of the SBOM size.
The proving system operates on the \textbf{Pasta curves} (Pallas and Vesta), leveraging their efficient cycle properties for recursive verification.

\paragraph{Circuit Logic}
To implement our zero-knowledge proofs, we use \textbf{Circom}~\cite{circom}, a domain-specific language with its associated toolchain (shown in Figure~\ref{fig:zkp.system.implementation}). Each statement in our system is expressed as a Circom \emph{circuit template}, which defines constraints on inputs and intermediate values. These templates are compiled into an intermediate representation tailored for zkSNARK systems that captures the circuit logic and enables efficient proof generation and verification. In Circom, computations are represented as circuits, where an arithmetic circuit is composed of \emph{signals} representing inputs, outputs, and intermediate values, and gates that perform arithmetic operations such as addition and multiplication.
The core of our design is the \texttt{Attest} circuit template. This circuit takes as private inputs the dependency attributes (name, version, license information, code checksum, permission, etc.), the hash from the previous step, randomizers ($\rho$) for uniqueness, and the Merkle path for the inclusion proof. The circuit computes a Poseidon hash of the dependency attributes, verifies its inclusion in a Merkle tree, and outputs an updated accumulator hash derived from the current hash and  $\rho$. 

We employ the \textbf{Poseidon} hash function (via \texttt{circomlib}) for all in-circuit commitments (PT and ST), as it minimizes the R1CS constraint count compared to traditional hashes.
The final compressed proof $\pi$ is generated using a \textbf{Spartan} SNARK wrapper, ensuring a constant proof size regardless of the SBOM depth.

\subsection{Case Studies}
We simulated \veribom\ against four distinct scenarios ($S_x$ illustrated in Table \ref{tab:scenarios} of the software supply chain lifecycle. \textit{Generation Time} includes witness computation and ZK proving. \textit{Verification Result} indicates the final policy enforcement decision.

\begin{table}[h]
    \centering
    \caption{\textbf{Experimental Validation of Lifecycle Scenarios.}}
    \label{tab:scenarios}
    \resizebox{\columnwidth}{!}{%
    \begin{tabular}{@{}llcccl@{}}
    \toprule
    \textbf{ID} & \textbf{Scenario} & \textbf{Artifact State} & \textbf{Authority Roots} & \textbf{Outcome} & \textbf{Failure Cause / Implication} \\ \midrule
    
    \textbf{$S_1$} & \textbf{Baseline Compliance} & Compliant SBOM & $\digest_{PT}, \digest_{ST}$ & \textcolor{ForestGreen}{\textbf{Valid}} & Authorized deployment. \\ 
    
    \textbf{$S_2$} & \textbf{Injection} & \makecell[l]{Unknown Package OR \\ Policy Violation} & $\digest_{PT}, \digest_{ST}$ & \textcolor{BrickRed}{\textbf{Invalid}} & \makecell[l]{Blocked by PT (Not found) OR \\ Blocked by ST (Flagged 0).} \\ 
    
    \textbf{$S_3$} & \textbf{Revocation} & Vulnerable SBOM & \makecell[c]{$\digest_{PT}$ (Same) \\ $\digest'_{ST}$ (\textbf{Updated})} & \textcolor{BrickRed}{\textbf{Invalid}} & \makecell[l]{ST Mismatch. \\ auditor flagged component as unsafe.} \\ 
    
    \textbf{$S_4$} & \textbf{Remediation} & Patched SBOM & $\digest_{PT}, \digest'_{ST}$ & \textcolor{ForestGreen}{\textbf{Valid}} & Vendor synced with new policy. \\ \bottomrule
    \end{tabular}%
    }
\end{table}
In \textit{$S_1$}, we consider the desirable case in which the package manager publishes root $\digest_{PT}$ and the auditor publishes root $\digest_{ST}$. The software vendor generates an SBOM for its software and generates a proof $\pi$, against the roots.
By the successful validation of the proof, the software vendor demonstrates to include in its software artifact the official dependencies that result \textit{compliant} with the policy constraints.

In \textit{$S_2$}, we consider an \textit{injection} case in which a software developer generates a proof $\pi*$ with dependencies no longer available within the package repository and/or dependencies that violate the policy constraints.
The software vendor results in a general status of non-compliance, but the absence of a plaintext SBOM preserve the confidentiality and does not allow any malicious actor to find the exploitable component.

In \textit{$S_3$}, we consider a \textit{revocation} case where the auditor has flagged a dependency as vulnerable while still used by the software vendor. The auditor hence updates its root as $\digest'_{ST}$.
The original proof $\pi$, generated against $\digest_{PT}$ and $\digest_{ST}$ is immediately rejected by the verifier, effectively blocking the now-stale artifact without client-side scanning.
This demonstrates a case of compliance with the package manager but of non-compliance vendor with the policy constraints.

In \textit{$S_4$}, we consider a \textit{remediation} case where the vendor responds to the incident in $S_3$. 
The vendor patch the software artifact by removing the non-compliant dependencies and regenerate the proof $\pi'$ against the updated $\digest_{PT}$ and $\digest'_{ST}$.

\subsection{Results}
We analyze the scalability of the proposed method along five dimensions: VC size, Merkle path size, proof size, proof generation time, and proof verification time. We evaluate each metric with respect to two variables: (1) the total number of packages in the system and (2) the number of dependencies associated with the software being audited. In our evaluation, the term \emph{dependencies} refers to the complete set of both immediate and transitive dependencies of the target library.
We run our experiments on a commodity laptop with an Intel(R) Core(TM) i5-8265U CPU processor with 16GB of RAM running Ubuntu~22.04.

  \begin{figure}[t]
	\centering
	\includegraphics[width=\linewidth]{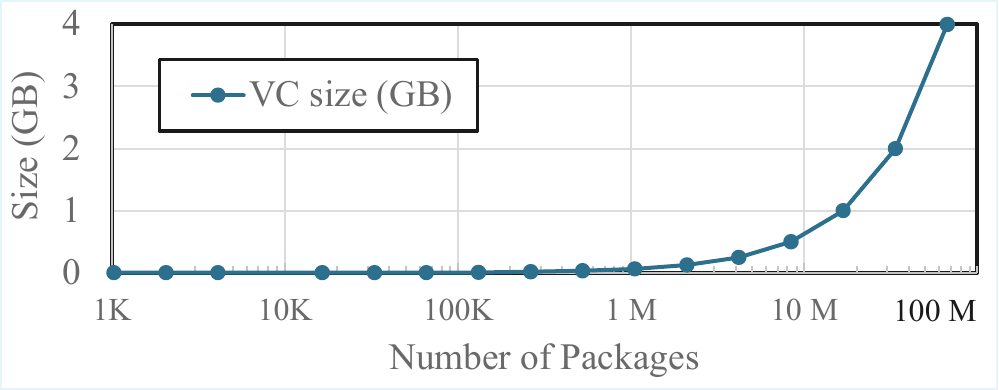}
	\caption{Storage for the package repository VC.}
	\label{fig:scalability_vc_size}
\end{figure}

\begin{figure}[t]
	\centering
	\includegraphics[width=\linewidth]{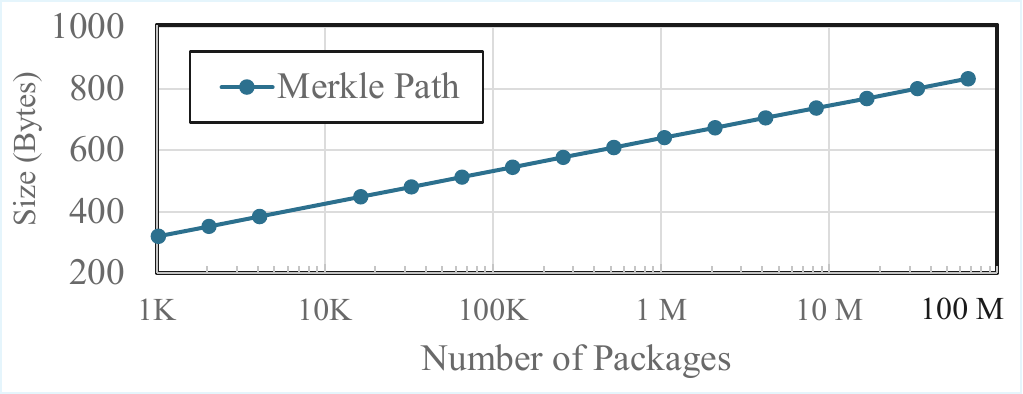}
	\caption{Merkle path size.}
	\label{fig:scalability_merkle_path_size}
\end{figure}

Figure~\ref{fig:scalability_vc_size} represents the additional public storage overhead introduced by the vector commitment over the package manager's metadata database. The x-axis is shown in logarithmic scale and indicates the total number of packages. The results demonstrate that for the current scale of ecosystem registries such as \texttt{crates.io}, the additional storage required by the vector commitment remains below 1~GB, even when committing to more than 10 million packages. Furthermore, the system is future-proof: even if the number of published packages exceeds 100 million, the total VC storage overhead would be approximately 4~GB which is within the capabilities of modern infrastructure.

Figure~\ref{fig:scalability_merkle_path_size}  shows the size of the inclusion proofs required to authenticate each dependency of an artefact.
The results confirm that the Merkle path size, and consequently, the complexity of proving authenticity for a single dependency grows logarithmically with the number of packages. For example, increasing the total number of packages from approximately 1 million to 2 million results in only one additional hash computation in the Merkle path. This demonstrates the efficiency and scalability of Merkle-based proofs, even as the ecosystem grows to tens or hundreds of millions of entries.
Furthermore, the proof size measures the storage overhead of the zero-knowledge proofs appended to the SBOM data. Our evaluation shows that the proof size remains nearly constant in all test scenarios, approximately 13\,KB regardless of the number of dependencies or the overall size of the package registry. This consistency is achieved by using a folding-based zkSNARK approach to generate proofs with bounded size.

Figure~\ref{fig:scalability_proof_time} illustrates the computational complexity incurred by the \prover\ during proof construction. 
The results show that the total number of packages has negligible impact on proof generation time, e.g., proving the authenticity of a library with 300 dependencies takes only about 4 seconds longer in a system with 1 million packages compared to one with 1{,}000 packages, representing less than a 10\% increase. 
By contrast, the number of dependencies has a higher effect: proving time scales approximately linearly with the number of dependencies. This is expected, as each dependency must be authenticated individually through a Merkle inclusion proof and rule check. Nonetheless, the generation times remain within practical bounds, confirming that the system is scalable and suitable for real-world usage scenarios involving complex software systems.

Moreover, Figure~\ref{fig:scalability_verify_time} shows the time taken by the verifier to validate the generated zero-knowledge proofs, which remains almost constant between 80 and 95 milliseconds across all evaluated configurations, regardless of the total number of packages or the number of dependencies in the software being verified. The slight variations are mainly due to external factors, such as system I/O or runtime overhead, rather than the intrinsic complexity of the verification process itself. This consistency is due to using succinct proof systems, i.e., as folding-based zkSNARKs. This property ensures that the system can scale effectively while maintaining low verification costs, making it suitable for integration in performance-sensitive environments, including client-side verifiers or automated software production pipelines.



\begin{figure}[t]
    \centering
    \includegraphics[width=\linewidth]{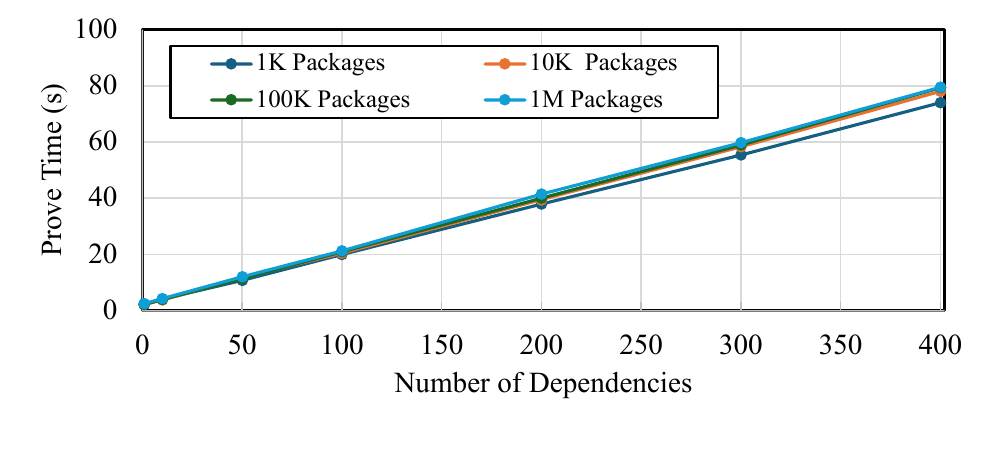}
    \caption{Proof generation time.}
    \label{fig:scalability_proof_time}
\end{figure}

\begin{figure}[t]
    \centering
    \includegraphics[width=\linewidth]{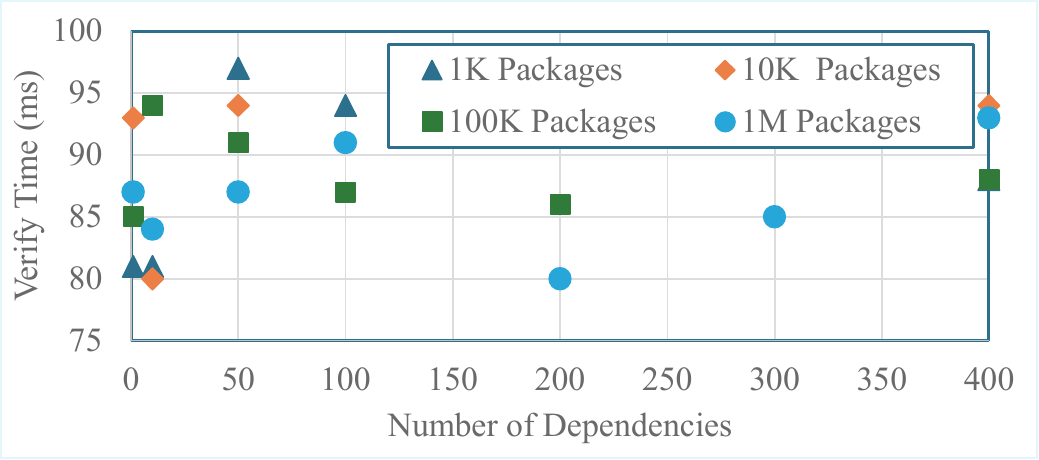}
    \caption{Verification time.}
    \label{fig:scalability_verify_time}
\end{figure}



\subsection{Discussion}
Our scalability experiments demonstrate that the proposed approach scales in all evaluated dimensions, indicating its practicality for real-life systems with small performance overhead.  
Moreover, a key advantage of our method is its potential for integration with existing industry standards and workflows. This is achieved by leveraging the current package databases maintained by package managers (e.g., \texttt{Crates.io} and \texttt{Homebrew}) without requiring modifications to their structure or content. All cryptographic mechanisms introduced in our approach operate as an \emph{external layer} over unmodified metadata (as illustrated in Fig.~\ref{fig:sys}). Therefore, our approach requires no changes to existing package registries or consumer tools, making it possible to integrate with state-of-the-art software supply chain infrastructures.  
In addition, our framework is language-agnostic and supports projects that span multiple software ecosystems. To demonstrate this, we conducted experiments involving Rust and C, showing that the framework can generate and verify proofs for libraries with cross-language dependencies. Table~\ref{tab:interoperability} summarizes the results for this interoperability scenario.

\begin{table}[h]
    \centering
    \caption{Interoperability}
    \label{tab:interoperability}
    \begin{adjustbox}{width=\linewidth}
    \begin{tabular}{|cc|cc|c|c|c|}
        \hline
        \multicolumn{2}{|c|}{Packages} & \multicolumn{2}{c|}{Dependencies} & \multirow{2}{*}{ \STAB{c}{Proof \\ Size}} & \multirow{2}{*}{ \STAB{c}{Proof \\ Time}} & \multirow{2}{*}{ \STAB{c}{Verify \\ Time}} \\ \cline{1-4}
        \multicolumn{1}{|c|}{C} & Rust & \multicolumn{1}{c|}{C} & Rust &  &  &  \\ \hline
        \multicolumn{1}{|c|}{\multirow{4}{*}{\STAB{c}{$2^{14}$\\ \texttt{Brew} \\ packages }}} & \multirow{4.1}{*}{\STAB{c}{$2^{14}$\\ \texttt{crates.io} \\ packages }} & \multicolumn{1}{c|}{1} & 10 & 13 KB & 4.1 s & $<$ 0.1 s \\ \cline{3-7} 
        \multicolumn{1}{|c|}{} &  & \multicolumn{1}{c|}{10} & 1 & 13 KB & 4.1 s & $<$ 0.1 s \\ \cline{3-7} 
        \multicolumn{1}{|c|}{} &  & \multicolumn{1}{c|}{50} & 10 & 13 KB & 13.6 s & $<$ 0.1 s \\ \cline{3-7} 
        \multicolumn{1}{|c|}{} &  & \multicolumn{1}{c|}{10} & 50 & 13 KB & 13.6 s &  $<$ 0.1 s \\ \hline
        \end{tabular}
    \end{adjustbox}
\end{table}

Our evaluation shows that the proposed system enables verifiers to check the authenticity of disclosed packages and their compliance with specified policies while preserving the confidentiality of sensitive SBOM data. Detecting compromised or policy‑violating dependencies is, in effect, guaranteed by our construction: zero‑knowledge proofs over vector commitments ensure that any disclosed dependency or policy claim is cryptographically enforced and cannot be forged. Consequently, additional controlled experiments aimed at revalidating correctness would yield limited new insight, as these guarantees derive from the underlying cryptographic foundations. 
Assessing how the system behaves in real‑world settings, where many independent parties publish and verify SBOMs, requires deployment at scale. Such a study would clarify practical factors including performance in distributed environments, coordination among participants, and adoption barriers. We leave this to future work as an important step toward studying the framework’s real‑world impact on software supply chain security.

Another consideration is the update mechanism for package repositories. When packages are added or modified, a new vector commitment is constructed to reflect the updated metadata. Although this requires recomputation, the cost remains low because the update affects only the path from the new node to the root, which is logarithmic in size. It is important to note that proofs are generated with respect to a specific version of the vector commitment, and previously generated proofs remain valid for the older version.
Future work will investigate incremental commitment schemes to further optimize update efficiency without compromising transparency or verifiability.  

Furthermore, all operations introduced in our approach, including VC construction and ZKP generation, can be performed by any party without requiring private keys or reliance on a trusted authority. This property facilitates adoption and decentralized implementation, i.e., multiple parties can collaboratively construct the VC and generate proofs, as any entity with access to public package metadata can reproduce commitments and proofs in a verifiable manner.

\subsection{Threats to Validity and Limitations}
The threats to the validity of our results and the limitations of the current design are as follows.

\begin{itemize}[leftmargin=*]
    \item \emph{Correctness of SBOMs}
     We assume that the {SBOM generated by the software vendor is correct and complete}, e.g., a software vendor does not omit a specific and vulnerable package from the SBOM while still including it in the delivered software. In such case, the generated proof would be correct and verifiable, but the software would be remain vulnerable.
    Therefore, \veribom\ still rely, like any technological component, on the correctness of the inputs provided by the software vendors.

    \item \emph{\prover\ performance overhead}
    Constructing Zero-Knowledge proofs entails a higher computational cost compared to other solutions like symmetric encryption.
    For extremely large SBOMs, the proof generation time might become expensive, however, our adoption of folding schemes specifically addresses this by reducing the complexity of incremental adding components to the proof.

    \item \emph{Dependency on trusted third parties}
    Our model relies on public information generated by trusted third parties.
    \veribom\ trusts $\digest_{PT}$ and $\digest_{ST}$, but if they are not updated with respect to recent security discoveries or compromised, \veribom\ will generate a valid proof for a potentially vulnerable artifact. 
    The solution is anyway flexible for allowing verification by checking multiple authorities simultaneously, reducing reliance on a single point of failure.

    \item \emph{Circuit correctness}
    The security of \veribom\ depends on the correctness of the arithmetic zero-knowledge circuit constraints. An under-constrained circuit or an incorrect circuit design could theoretically allow a malicious \prover\ to forge compliance. Therefore, we propose the use of formal verification tools for zero-knowledge circuits as future work to mathematically guarantee constraint correctness.
\end{itemize}